\documentclass[a4paper,UKenglish,cleveref, autoref, thm-restate]{lipics-v2021}

\pdfoutput=1 
\hideLIPIcs 

\usepackage{amsmath,amsthm,amsfonts}
\usepackage{xcolor, xspace}
\usepackage[ruled,linesnumbered,vlined]{algorithm2e}
    \SetAlgoSkip{}
    \DontPrintSemicolon
    \SetVlineSkip{1pt}
\usepackage[nolist]{acronym}
\usepackage{tikz}
\usepackage{mathtools}
\usepackage{dsfont}
\usepackage{pgf}
\usepackage{bbm}  

\graphicspath{{figures/}}
\nolinenumbers 

\crefname{algocf}{Algorithm}{Algorithms}
\Crefname{algocf}{Algorithm}{Algorithms}
\SetKwComment{Comment}{$\triangleright$\ }{}
\SetCommentSty{itshape}

\usepackage{multicol}
\usepackage{cleveref}
\newcounter{caseCounter}
\theoremstyle{remark}
\newtheorem{case}[caseCounter]{Case}

\crefname{case}{case}{cases}
\Crefname{case}{Case}{Cases}
\newcounter{propertyCounter}
\theoremstyle{remark}

\crefname{prop}{property}{properties}
\Crefname{prop}{Property}{Properties}
\newcounter{obsCounter}
\theoremstyle{remark}
\newtheorem{obs}[obsCounter]{Observation}

\crefname{obs}{Observation}{Observations}
\Crefname{obs}{Observation}{Observations}
\newcommand{\UNE}{{\textsc{une}}}
\newcommand{\UW}{{\textsc{uw}}}
\newcommand{\USE}{{\textsc{use}}}
\newcommand{\N}{{\textsc{n}}}
\newcommand{\NW}{{\textsc{nw}}}
\newcommand{\SW}{{\textsc{sw}}}
\renewcommand{\S}{{\textsc{s}}}
\newcommand{\SE}{{\textsc{se}}}
\newcommand{\NE}{{\textsc{ne}}}
\newcommand{\DNW}{{\textsc{dnw}}}
\newcommand{\DSW}{{\textsc{dsw}}}
\newcommand{\DE}{{\textsc{de}}}
\newcommand{\X}{{\textsc{x}}}

\newcommand{\tiled}{\mathcal{T}}
\newcommand{\bb}{{\mathfrak{B}}}
\newcommand{\bc}{{\mathfrak{C}}}
\renewcommand{\O}{\mathcal{O}}

\title{Efficient Shape Formation by 3D Hybrid Programmable Matter: An Algorithm for Low Diameter Intermediate Structures}
\titlerunning{Efficient Shape Formation by 3D Hybrid Programmable Matter} 
\author{Kristian {Hinnenthal}}{Paderborn University, Germany \and \url{https://www.uni-paderborn.de/person/32229} }{kristian.hinnenthal@hsbi.de}{https://orcid.org/0000-0001-9464-295X}{}
\author{David {Liedtke}}{Paderborn University, Germany \and \url{https://www.uni-paderborn.de/person/55557} }{liedtke@mail.upb.de}{https://orcid.org/0000-0002-4066-0033}{}
\author{Christian {Scheideler}}{Paderborn University, Germany \and \url{https://cs.uni-paderborn.de/ti/personal/prof-dr-rer-nat-christian-scheideler} }{scheideler@upb.de}{https://orcid.org/0000-0002-5278-528X}{}
\authorrunning{K. Hinnenthal, D. Liedtke, C. Scheideler} 
\Copyright{Kristian Hinnenthal and David Liedtke and Christian Scheideler} 
\funding{This work was supported by the DFG Project SCHE 1592/10-1.}

\EventEditors{John Q. Open and Joan R. Access}
\EventNoEds{2}
\EventLongTitle{42nd Conference on Very Important Topics (CVIT 2016)}
\EventShortTitle{CVIT 2016}
\EventAcronym{CVIT}
\EventYear{2016}
\EventDate{December 24--27, 2016}
\EventLocation{Little Whinging, United Kingdom}
\EventLogo{}
\SeriesVolume{42}
\ArticleNo{23}

\begin{document}

  \maketitle

  \begin{abstract}
    This paper considers the shape formation problem within the 3D hybrid model, where a single agent with a strictly limited viewing range and the computational capacity of a deterministic finite automaton manipulates passive tiles through pick-up, movement, and placement actions. 
    The goal is to reconfigure a set of tiles into a specific shape termed an \emph{icicle}. 
    The icicle, identified as a dense, hole-free structure, is strategically chosen to function as an intermediate shape for more intricate shape formation tasks.
    It is designed for easy exploration by a finite state agent, enabling the identification of tiles that can be lifted without breaking connectivity.
    Compared to the line shape, the icicle presents distinct advantages, including a reduced diameter and the presence of multiple removable tiles.
    We propose an algorithm that transforms an arbitrary initially connected tile structure into an icicle in $\mathcal{O}(n^3)$ steps, matching the runtime of the line formation algorithm from prior work.
    Our theoretical contribution is accompanied by an extensive experimental analysis, indicating that our algorithm decreases the diameter of tile structures on average.
  \end{abstract}

\section{Introduction}
\label{sec:intro}

Advancements in molecular engineering have led to the development of a series of computing DNA robots designed for nano-scale operations.
These robots are intended to perform simple tasks such as transporting cargo, facilitating communication, navigating surfaces of membranes, and pathfinding~\cite{Anupama2017DNARobotCargo, Akter2022DNARobotCargoCollective, Li2021DNARobotWalkMembrane, Chao2019DNARobotPathfinding}.
Envisioning the future of nanotechnology, we anticipate a scenario where a collective of computing particles collaboratively acts as programmable matter -- a homogeneous material capable of altering its shape and physical properties programmably.
There are numerous potential applications:
For environmental remediation, particles may construct nanoscale filtration systems to remove pollutants from air or water. 
They may also be deployed within the human body to construct intricate structures for targeted drug delivery, perform nanoscale surgeries, or repair damaged tissues at a cellular level. 
Additionally, they could assemble nanoscale circuits and components, enabling the development of more efficient and compact electronic devices.
Each of those scenarios is an application of the \emph{shape formation problem}, which is the subject of this paper.

Over the past few decades, various models of programmable matter have emerged, primarily distinguished by the activity of entities within them.
Passive systems consist of entities (called tiles) that undergo movement and bonding exclusively in response to external stimuli, such as current or light, or based on their inherent structural properties, such as specific glues on the surfaces of tiles. 
Examples of these passive systems include the DNA tile assembly models aTAM, kTAM, and 2HAM, which are extensively discussed in the survey~\cite{Patitz2014TileAssemblySurvey}, as well as population protocols \cite{Angluin2006PopulationProtocols}, and slime molds \cite{Chen2023SlimeMoldSurvey}.
In contrast, active systems consist of entities (called particles, agents or robots) that independently perform computation and movement to accomplish tasks. 
Notable examples encompass the Amoebot model~\cite{Derakhshandeh2014Amoebot}, modular self-reconfigurable robots~\cite{Tan2020ModularRobots,Tucci2018ModularRobots}, the nubot model~\cite{Woods2013NuBotModel}, metamorphic robots~\cite{Chirikjian1994MetamorphicRobots,Walter2004MetamorphicRobots}, and swarm robotics~\cite{Werfel2014SwarmRobots}.

While fabricating computing DNA robots remains challenging, producing simple passive tiles from folded DNA strands is efficient and scalable~\cite{HeuerJungemann2019DNATileSurvey}.
The hybrid model of programmable matter~\cite{Gmyr2018Hybrid2DShapeRecognition, Gmyr2020Hybrid2DShapeFormation, Hinnenthal2020Hybrid3D, Nokhanji2022Hybrid2DDynamicLine, Kostitsyna2023Hybrid3DCoating} offers a compromise between feasibility and utility.
This model involves a small number of active agents with the computational capabilities of deterministic finite automata together with a large set of passive building blocks, called tiles.
Agents can manipulate the structure of tiles by picking up a tile, moving it, and placing it at some spot.
A key advantage of the hybrid approach lies in the reusability of agents upon completing a task, where in purely active systems, particles become part of the formed structure.

In this paper, we address the shape formation problem within the 3D hybrid model, with the ultimate goal of transforming an arbitrary initial arrangement of tiles into a predefined shape.
We consider tiles in the shape of rhombic dodecahedra, i.e., polyhedra featuring 12 congruent rhombic faces, positioned at nodes within the adjacency graph of face-centered cubic (FCC) stacked spheres (see \cref{subfig:icicle}).
Unlike rectangular tiles, the rhombic dodecahedron presents a distinct advantage:  it allows an agent to orbit around a tile without risking connectivity.
This property is particularly valuable in liquid or low gravity environments, where it prevents unintended separation between the agents and the tiles.

Achieving universal 3D shape formation faces a key challenge: identifying tiles that can be lifted without disconnecting the tile structure (referred to as \emph{removable} tiles).
Even if such tiles exist, locating them requires exploring the tile structure, demanding $\Omega(D \log(\Delta))$ memory bits for graphs with a diameter $D$ and degree $\Delta$ \cite{Fraigniaud2004GraphExploration}.
When limited to constant memory, navigating plane labyrinths requires two placeable markers (pebbles)~\cite{Hoffmann1982OnePebbleDoesNotSuffice,Blum1978OnThePowerOfTheCompass}.
In the 2D context, finding removable tiles is impossible without prior modification of the tile structure, as discussed in \cite{Gmyr2020Hybrid2DShapeFormation}.
In 3D, complexity increases significantly, with instances where any tile movement can locally disconnect the structure.
As discussed above, the agent is unable to verify whether this disconnection also occurs globally.
To address these challenges, we make the assumption that the agent carries a tile initially, using it to uncover removable tiles through successive tile movements. 
It is still entirely unclear whether otherwise a removable tile can be found in all 3D instances.
For that reason, our primary goal is to construct an intermediate structure that is easily navigable by constant-memory agents and allows the identification of removable tiles without relying on an initially carried~tile.

\subsection{Our Contribution} \label{subsec:contribution}

The intermediate structure we propose is termed an \emph{icicle}, characterized by a platform representing a parallelogram and downward-extending lines of tiles from the platform (see \cref{subfig:icicle}).
We~present a single-agent algorithm that transforms any initially connected tile structure into an icicle in $\O(n^3)$ steps, matching the efficiency of the line formation algorithm from prior work~\cite{Hinnenthal2020Hybrid3D}.
While both the icicle and the line enable agents without an initial tile to find removable tiles, the icicle presents distinct advantages.
In the best-case scenario, the diameter $D$ of an icicle can be as low as $\O(n^{\frac{1}{3}})$, whereas a line consistently maintains a diameter of $n$.
Furthermore, an icicle encompasses multiple removable tiles, which removes the necessity to traverse the intermediate shape completely to locate a removable tile.
Our paper includes comprehensive simulation results, indicating that, on average, our algorithm reduces the diameter of the tile structure.
In addition, the runtime observed in the simulations consistently falls below the bound established in our runtime analysis. 
Across all~simulations, the runtime remains well within the vicinity of $n^2$. 
It is noteworthy that we identified an edge case where the diameter could increase by a factor of $\mathcal{O}(n^{\frac{1}{3}})$, although we believe this to be the worst-case.

\subsection{Related Work} \label{subsec:related}

The 3D variant of the hybrid model was introduced in \cite{Hinnenthal2020Hybrid3D}, where the authors presented an algorithm capable of transforming any connected input configuration into a line in $\O(n^3)$ steps.
In \cite{Kostitsyna2023Hybrid3DCoating}, the authors address the coating problem, providing a solution that solves the problem in worst-case optimal $\O(n^2)$ steps.
They assume a single active agent that has access to a constant number of distinguishable tile types.

Significant progress has been made in recent years regarding the 2D version of the hybrid model.
For instance, in~\cite{Gmyr2020Hybrid2DShapeFormation}, the authors address the 2D shape formation problem, presenting algorithms for a single active agent that efficiently constructs line, block, and tree structures - each being hole-free structures with specific advantages and disadvantages — in worst-case optimal $\O(n^2)$ steps.
Another publication, \cite{Gmyr2018Hybrid2DShapeRecognition}, explores the recognition of parallelograms with a specific height-to-length ratio.
The most recent publication \cite{Nokhanji2022Hybrid2DDynamicLine} solves the problem of maintaining a line of tiles in presence of multiple agents and dynamic failures of the tiles.

Closely tied to the hybrid model is the well-established Amoebot model, where computing particles traverse an infinite triangular lattice through expansions and contractions.
In \cite{Derakhshandeh2015AmoebotShapeFormation}, the authors showcase the construction of simple shapes like hexagons or triangles within the Amoebot model.
Expanding on this work, \cite{Derakhshandeh2016AmoebotShapeFormation} introduces a universal shape formation algorithm capable of constructing an arbitrary input shape using a constant number of equilateral triangles, with the scale depending on the number of amoebots.
Notably, this work assumes common chirality, a sequential activation schedule, and randomization.
Subsequent improvements are presented in \cite{DiLuna2020AmoebotShapeFormation}, where a deterministic algorithm is introduced, enabling amoebots to form any Turing-computable shape without the need for common chirality or randomization.
In \cite{Kostitsyna2022AmoebotFaultTolerantSF}, the authors consider shape formation in the presence of a finite number of faults, where a fault resets an amoebot's memory.
They solve the hexagon formation problem, assuming the existence of a fault-free leader.
A recent extension of the Amoebot model, discussed in \cite{Padalkin2023CircuitShapeFormation}, considers joint movements of Amoebots.
The authors simulate various shape formation algorithms as a proof of concept.

In both \cite{Gmyr2020Hybrid2DShapeFormation} and this paper, shape formation algorithms are introduced that construct an intermediate shape, intended to serve as the foundation for more advanced shape formation algorithms.
A similar strategy is explored in \cite{Hurtado2015ReconfModularRobots}, where 2D lattice-based modular robots initially transform into a canonical shape before achieving the final desired shape.
An approach that does not rely on canonical intermediate structures is considered in \cite{Peters2023FastReconfiguration}.
The authors present primitives for the Amoebot model that establish shortest path trees within the amoebot structure and subsequently directly route amoebots to their target position.

The concept of shape formation is extensively studied in the field of modular robotics and metamorphic robots, often referred to as self-reconfiguration.
A comprehensive survey on this topic can be found in \cite{Thalamy2019ShapeFormationModularRobotsSurvey}.
In the field of swarm robotics, shape formation is often closely related to the problem of computing collision-free paths \cite{Wang2020SwarmRoboticsShapeFormation,Li2019SwarmRoboticsShapeFormation}.

\subsection{Model Definition} \label{subsec:model}

\begin{figure}[!t]
    \noindent
        \begin{subfigure}[b]{.4\linewidth} %
            \centering%
            \includegraphics[width=0.95\linewidth,page=1]{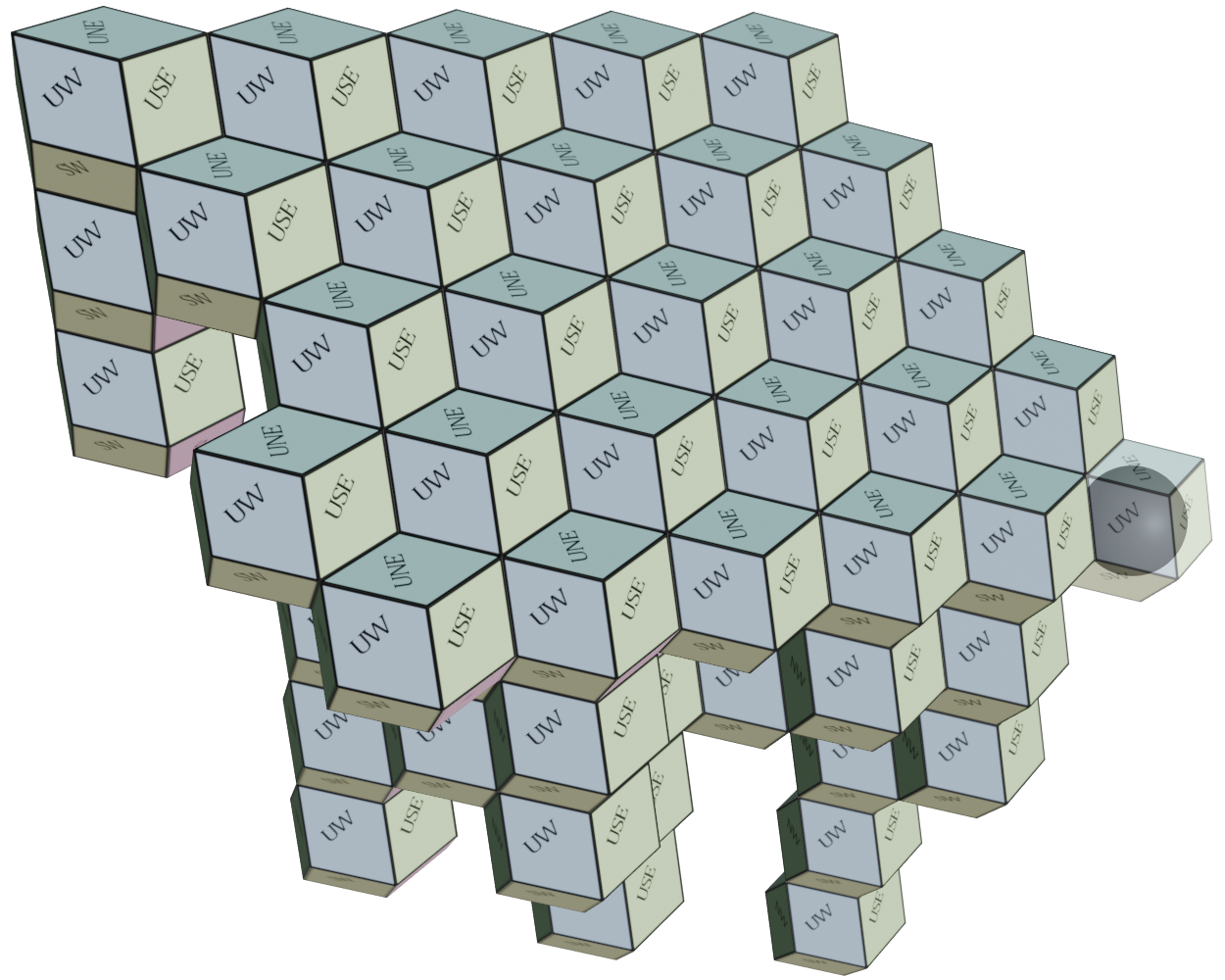}%
            \subcaption{}%
            \label{subfig:icicle}%
        \end{subfigure}%
        \begin{subfigure}[b]{0.192\linewidth}
            \centering%
            \includegraphics[width=0.95\linewidth,page=1]{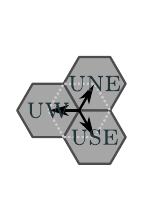}%
            \subcaption{}
            \label{subfig:compassDirections1}
        \end{subfigure}%
        \begin{subfigure}[b]{0.192\linewidth}
            \centering%
            \includegraphics[width=0.95\linewidth,page=2]{directionsAndAxes}%
            \subcaption{}
            \label{subfig:compassDirections2}
        \end{subfigure}
        \begin{subfigure}[b]{0.192\linewidth}
            \centering%
            \includegraphics[width=0.95\linewidth,page=3]{directionsAndAxes}%
            \subcaption{}
            \label{subfig:compassDirections3}
        \end{subfigure}%
    \caption{(a) An example configuration that has the shape of an icicle; the agent (depicted as a sphere) is positioned at a tiled node within the platform representing a parallelogram. (b--d) The twelve compass directions divided into upwards (b), plane (c) and downwards directions (d).}
    \label{fig:compassDirections}
\end{figure}

We consider a single active agent $r$ with limited sensing and computational power that operates on a finite set of passive \emph{tiles} positioned at nodes of some specific underlying graph $G$, which we define in the following.
Consider the close packing of equally sized spheres at each point of the infinite face-centered cubic lattice.
Let $G = (V,E)$ be the adjacency graph of spheres in that packing, and consider an embedding of $G$ in $\mathbb{R}^3$ in which all edges have equal length, e.g., the trivial embedding where the edge length equals the radius of the spheres. 
Cells in the dual graph of $G$ w.r.t.\ that embedding have the shape of rhombic dodecahedra, i.e., polyhedra with 12 congruent rhombic faces (see \cref{subfig:icicle}).
This is also the shape of every cell in the Voronoi tessellation of $G$, i.e., that shape completely tessellates 3D space.
Consider a finite set of tiles that have the shape of rhombic dodecahedra.
Tiles are passive, in the sense that they cannot perform any computation or movement on their own.
A node $v \in V$ is \emph{tiled}, if there is a passive tile positioned at $v$; otherwise node $v$ is \emph{empty}.
Each node can hold at most one tile and each tile is placed at at most one node at a time.
Each node in $V$ has precisely twelve neighbors whose relative positions are described by the twelve compass directions $\UNE,\UW,\USE,\N,\NW,\SW,\S,\SE,\NE,\DNW,\DSW$ and $\DE$ (see \cref{subfig:compassDirections1,subfig:compassDirections2,subfig:compassDirections3}).
Take note that $G$ contains infinitely many copies of the infinite triangular lattice, which serves as the underlying graph in the 2D variant.
This allows us to visually depict 3D examples as a stack of 2D hexagonal tiles, as shown in \cref{fig:compassDirections}.

A \emph{configuration} $C = (\tiled{},p)$ is the set $\tiled{}$ that contains all tiled nodes together with the agent's position $p$.
We call $C$ \emph{connected}, if $G|_{\tiled{}}$ is connected or if $G|_{\tiled{} \cup \{p\}}$ is connected and the agent carries a tile, where $G|_{W}$ denotes the subgraph of $G$ induced by some nodeset $W$.
That is, we allow the subgraph induced by all tiled nodes to disconnect, as long as a tile carried by the agent maintains connectivity.
This constraint prevents the agent and tiles to drift apart, e.g., in liquid or low gravity environments.

The agent $r$ is the only active entity in this model.
It has strictly limited sensing and computing power and can act on passive tiles by picking up a tile, moving and placing it at some spot.
Particularly, we assume an agent with the computational capabilities of a deterministic finite automaton that performs discrete steps of \emph{Look-Compute-Move} cycles.
In the \emph{look}-phase, the agent observes whether its current position $p$ and the twelve neighbors of $p$ are tiled or empty.
The agent is equipped with a compass that allows it to distinguish the relative positioning of its neighbors using the twelve above mentioned compass directions.
Its initial rotation and chirality can be arbitrary, but we assume that it remains consistent throughout the execution.
For ease of presentation, our algorithms and their analysis are described according to the robot's local view, i.e., we do not distinguish between local and global compass directions.
Based on the information gathered in the look phase, the agent determines its next state transition according to the finite automaton in the \emph{compute}-phase.
In the \emph{move} phase, the agent performs an \emph{action} that corresponds to the prior state transition.
It either (i) moves to an empty or tiled node adjacent to $p$, (ii) places a tile at $p$, if $p \notin \tiled{}$ and $r$ carries a tile, (iii) picks up a tile from $p$, if $p\in \tiled{}$ and $r$ carries no tile, or (iv) terminates.
The agent can carry at most one tile at a time and during actions (ii) and (iii) the agent loses and gains a tile, respectively.
It's worth noting that we allow the agent to move through tiles while carrying one simultaneously.
From a practical standpoint, this capability can be facilitated by conceptualizing tiles as hollow and foldable.
It is assumed that the agent is initially positioned at a tiled node, as otherwise, there might be no valid action available.
Additionally, we assume that the agent initially carries a tile, a justification for which was provided in \cref{sec:intro}.
While the agent is technically a finite automaton, we describe algorithms from a higher level of abstraction textually and through pseudocode.
It is easy to see that a constant number of variables of constant-size domain each can be incorporated into the agent's constantly many states.


\subsection{Problem Statement} \label{subsec:problem}

Consider an arbitrary initially connected configuration $C^0 = (\tiled{}^0,p^0)$ with $p^0 \in \tiled{}$.
Superscripts in our notation generally refer to step numbers and may be omitted if they are clear from the context.
An algorithm solves the \emph{icicle formation problem}, if its execution results in a sequence of connected configurations $C^0 = (\tiled{}^0,p^0),\dots,C^{T} = (\tiled{}^0,p^0)$ such that nodes in $\tiled{}^{T}$ are in the shape of an icicle (which we define below), $C^{t}$ results from $C^{t-1}$ for $1\leq t \leq {T}$ by applying some action (i)--(iii) to $p^{t-1}$, and the agent terminates (iv) in step ${T}$.

For some node $v \in V$, we denote $v + \X$ the node that is neighboring $v$ in some compass direction $\X$ and $-\X$ the opposite compass direction of $\X$, e.g., $-\UNE = \DSW$.
We call a maximal consecutive array of tiles in direction $\N$ and $\S$ a \emph{column}, in direction $\NW$ and $\SE$ a \emph{row}, and in direction $\UNE$ and $\DSW$ a \emph{tower}.
A \emph{parallelogram} is a maximal consecutive array of equally sized columns $c_0,...,c_m$ (ordered from west to east) whose southernmost tiles at nodes $v_0,...,v_m$ are contained in the same row, i.e., $v_{i} + \SE = v_{i+1}$ for all $0 \leq i < m$.
In a \emph{partially filled} parallelogram, column $c_0$ can have smaller size than columns $c_1,...,c_m$.

An \emph{icicle} is defined as a connected set of towers whose uppermost tiles are contained within the same (partially filled) parallelogram, as illustrated in \cref{subfig:icicle}. 
In other words, tiles `grow' from a single uppermost parallelogram in the $\DSW$ direction, hence the chosen name `icicle'.
Notably, in an icicle, any tile with a neighboring tile at $\UNE$ but not at $\DSW$ (some locally $\DSW$-most tile below the parallelogram) can be picked up without violating connectivity (it is \emph{removable}).
If there is no such tile, i.e., all towers have size one, the northernmost tile of the westernmost column is removable.

\subsection{Structure of the Paper} \label{subsec:structure}

In \cref{sec:preliminaries}, we introduce essential terminology crucial for understanding the algorithm and its subsequent analysis. 
The non-halting icicle-formation algorithm is presented in \cref{sec:algorithm}. 
In \cref{sec:analysis}, we provide formal proofs establishing that the algorithm converges any initially connected configuration into an icicle. 
The termination criteria and a detailed analysis of its runtime are discussed in \cref{sec:runtime}.
Finally, \cref{sec:experimental} explores the simulation results.

\section{Preliminaries} \label{sec:preliminaries}

We assign $x,y$ and $z$ coordinates to each node $v \in V$, denoted by $c(v) = (x(v),y(v),z(v))$, where the $x$-coordinates grow from $\SE$ to $\NW$, $y$-coordinates from $\S$ to $\N$, and $z$-coordinates from $\DSW$ to $\UNE$ (see \cref{subfig:coordinateAxes}).
The coordinates transition between neighbors as follows:

\begin{obs}\label{eq:transition}
    Let $w$ be some reference node with $c(w) = (0,0,0)$. 
    The following holds:
    \setlength{\abovedisplayskip}{3pt}
    \setlength{\belowdisplayskip}{0pt}
    \[
        \arraycolsep=0pt
        \begin{array}{lrclrrrclrclrrrclrclrrr}
            c(w + \UNE&) &\;=&\;(&\phantom{-}0,&0,&1)  &\phantom{aaa}& c(w + \UW&)  &\;=&\;(&1,&-1,&1)   &\phantom{aaa}& c(w + \USE&) &\;=&\;(&0,&-1,&1)\\
            c(w + \N&)   &\;=&\;(&0,&1,&0)  &&  c(w + \NW&)  &\;=&\;(&1,&0,&0)  &&  c(w + \SW&)  &\;=&\;(&1,&-1,&0)\\
            c(w + \S&)   &\;=&\;(&0,&-1,&0)  &&  c(w + \SE&)  &\;=&\;(&-1,&0,&0)  &&  c(w + \NE&)  &\;=&\;(&-1,&1,&0)\\
            c(w + \DSW&) &\;=&\;(&0,&0,&-1)  &&  c(w + \DE&)  &\;=&\;(&-1,&1,&-1)  &&  c(w + \DNW&) &\;=&\;(&0,&1,&-1)
        \end{array}
    \]
\end{obs}

Given some nodeset $S$, let $x_{min}^S,x_{max}^S$ be the minimum and maximum $x$-coordinate of any node in $S$, and define $y_{min}^S,y_{max}^S, z_{min}^S$ and $z_{max}^S$ accordingly.
We normalize coordinates according to the minimum coordinates in the initial set of tiled nodes $\tiled{}^0$, i.e., we set we set $x_{min}^{\tiled{}^0} = y_{min}^{\tiled{}^0} = z_{min}^{\tiled{}^0} = 0$.
The \emph{bounding cylinder} $\bc(S)$ is the set of all nodes (both empty and tiled) whose coordinates are bounded by the minimum and maximum $x$- and $y$-coordinates in $S$, i.e., $\bc(S) = \{v \in V \mid x_{min}^S \leq x(v) \leq x_{max}^S, y_{min}^S \leq y(v) \leq y_{max}^S\}$ (see \cref{subfig:boundingBoxAndCylinder-1}).
Similarly, in the \emph{bounding box} $\bb(S)$ we further bound by the $z$-coordinate, i.e., $\bb(S) = \{v \in \bc(S) \mid z_{min}^S \leq z(v) \leq z_{max}^S\}$.
We refer to the extent of a bounding box along the $x$-, $y$- and $z$-axes as its \emph{width}, \emph{height}, and \emph{depth}.
Note that by the choice of our coordinate axes, the bounding box is always a filled (potentially degenerated) parallelepiped (a 3D rhomboid; see \cref{subfig:boundingBoxAndCylinder-2}).
A node $v$ is $\emph{inside}$ the bounding cylinder (box) of $S$, if $v \in \bc(S)$ ($v \in \bb(S)$); otherwise, $v$ is \emph{outside} of the bounding cylinder (box) of $S$.

\begin{figure}[!t]
    \centering
    \begin{subfigure}[b]{0.22\linewidth}
        \centering%
        \includegraphics[trim=0 -9.5mm 0 0,width=0.95\linewidth,keepaspectratio,width=0.8\linewidth,page=4]{directionsAndAxes}%
        \subcaption{}
        \label{subfig:coordinateAxes}
    \end{subfigure}%
    \begin{subfigure}[b]{0.36\linewidth}
        \centering%
        \includegraphics[width=0.9\linewidth,page=2]{boundingBoxAndCylinder}%
        \subcaption{}
        \label{subfig:boundingBoxAndCylinder-1}
    \end{subfigure}%
    \begin{subfigure}[b]{0.36\linewidth}
        \centering%
        \includegraphics[width=0.9\linewidth,page=1]{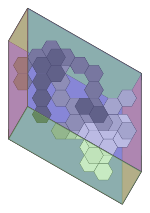}%
        \subcaption{}
        \label{subfig:boundingBoxAndCylinder-2}
    \end{subfigure}%
    \caption{Illustrating the $x$-, $y$-, and $z$-coordinate axes (a), the bounding cylinder (b), which infinitely extends in directions $\UNE$ and $\DSW$ as indicated by the arrows, and the bounding box (c) of an example configuration. Tiles are shaded according to their $z$-coordinate, with brighter shades representing lower $z$-coordinates. In the example, there is one layer that contains two fragments (darkest shade of gray), and four layers that each contain a single fragment.
    }
    \label{fig:boundingboxFilled}
\end{figure}

A \emph{layer} $L_i$ is the set of all nodes with $z$-coordinate $i$ that are contained in the bounding cylinder of all tiled nodes, i.e., $L_i = \{v \in \bc{}(\tiled{}) \mid z(v) = i\}$.
We refer to nodes with $z$-coordinate greater than and less than $i$ as the nodes \emph{above} and \emph{below} layer $L_i$, respectively.
The nodeset of a connected component of $G|_{L_i \cap \tiled{}}$ is called a \emph{fragment (of $L_i$)} (see \cref{fig:boundingboxFilled}).

\section{The Algorithm} \label{sec:algorithm}

From a high-level perspective, the agent iteratively transforms locally uppermost fragments into partially filled parallelograms.
This involves rearranging tiles within the same layer and, at times, positioning tiles below the current layer to ensure connectivity.
Whenever the agent encounters tiles of some layer above, it moves further upwards.
Once a parallelogram is successfully formed, the subsequent step entails its \emph{projection}. 
Essentially, during this projection, each tile in the fragment is shifted to the first empty node in the $\DSW$ direction.

\newcommand{\buildPar}{\textsc{BuildPar}}
\newcommand{\proj}{\textsc{Project}}
\newcommand{\icicleAlgo}{\textsc{BuildIcicle}}

In the following, we provide detailed textual descriptions of the parallelogram formation and projection procedures \buildPar{} and \proj{}, as well as the full icicle formation algorithm \icicleAlgo{}.
For completeness, their pseudocodes can be found in \cref{sec:appendix-pseudocode}. 

\subsection{A 2D Parallelogram Formation Algorithm} \label{subsec:algorithm-parallelogram}

\begin{figure}[t]
    \centering
    \begin{subfigure}[c]{0.2\linewidth}
        \centering%
        \includegraphics[width=0.95\linewidth,keepaspectratio,page=1]{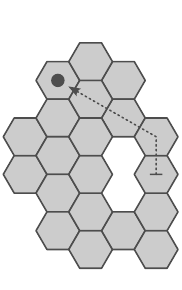}%
        \subcaption{}
        \label{subfig:parallelogram-a}
    \end{subfigure}%
    \hfill
    \begin{subfigure}[c]{0.2\linewidth}
        \centering%
        \includegraphics[width=0.95\linewidth,keepaspectratio,page=2]{parallelogram}%
        \subcaption{}
        \label{subfig:parallelogram-b}
    \end{subfigure}%
    \hfill
    \begin{subfigure}[c]{0.2\linewidth}
        \centering%
        \includegraphics[width=0.95\linewidth,keepaspectratio,page=3]{parallelogram}%
        \subcaption{}
        \label{subfig:parallelogram-c}
    \end{subfigure}%
    \hfill
    \begin{subfigure}[c]{0.2\linewidth}
        \centering%
        \includegraphics[width=0.95\linewidth,keepaspectratio,page=4]{parallelogram}%
        \subcaption{}
        \label{subfig:parallelogram-d}
    \end{subfigure}%
    \hfill
    \begin{subfigure}[c]{0.2\linewidth}
        \centering%
        \includegraphics[width=0.95\linewidth,keepaspectratio,page=5]{parallelogram}%
        \subcaption{}
        \label{subfig:parallelogram-e}
    \end{subfigure}%
    \linebreak
    \begin{subfigure}[c]{0.2\linewidth}
        \centering%
        \includegraphics[width=0.95\linewidth,keepaspectratio,page=6]{parallelogram}%
        \subcaption{}
        \label{subfig:parallelogram-f}
    \end{subfigure}%
    \hfill
    \begin{subfigure}[c]{0.2\linewidth}
        \centering%
        \includegraphics[width=0.95\linewidth,keepaspectratio,page=7]{parallelogram}%
        \subcaption{}
        \label{subfig:parallelogram-g}
    \end{subfigure}%
    \hfill
    \begin{subfigure}[c]{0.2\linewidth}
        \centering%
        \includegraphics[width=0.95\linewidth,keepaspectratio,page=8]{parallelogram}%
        \subcaption{}
        \label{subfig:parallelogram-h}
    \end{subfigure}%
    \hfill
    \begin{subfigure}[c]{0.2\linewidth}
        \centering%
        \includegraphics[width=0.95\linewidth,keepaspectratio,page=9]{parallelogram}%
        \subcaption{}
        \label{subfig:parallelogram-i}
    \end{subfigure}%
    \hfill
    \begin{subfigure}[c]{0.2\linewidth}
        \centering%
        \includegraphics[width=0.95\linewidth,keepaspectratio,page=10]{parallelogram}%
        \subcaption{}
        \label{subfig:parallelogram-j}
    \end{subfigure}%
    \caption{The parallelogram formation algorithm on a 2D configuration. The agent performs multiple steps between each depicted configuration. In (a) and (b) the agent finds a westernmost column, and in (j) the agent terminates. In all other cases, a tile is shifted from the cross to the circle, where the dashed lines indicate the path traversed before placing the tile. The path back to where the tile is picked up as well as the movement to the next column (e.g., (e)--(f)) is not shown.}
    \label{fig:parallelogram}
\end{figure}

Refer to \cref{fig:parallelogram} for an illustrative example of the algorithm in action.
The algorithm initiates with the agent searching for a locally westernmost column.
In configurations where multiple columns share the same $x$-coordinate and are locally westernmost, the agent prioritizes finding the northernmost among them. 
This is achieved by moving in the $\NW$, $\SW$, and $\N$ directions, prioritized in that order, until no more tile is encountered in any of these directions.
Eventually the agent stops upon reaching the northernmost tiled node $v$ of some column $c$.
We refer to the steps involved in finding column $c$ as the \emph{search phase}.

Subsequently, it executes the \buildPar{} procedure, which we describe in the following: 
Starting from node $v$, the agent traverses each column in the configuration from $\N$ to $\S$.
If, during the traversal of the first column $c$, the agent encounters either a more western column (as depicted in \cref{subfig:parallelogram-b}) or a column with the same $x$-coordinate as $c$ to the north while moving $\N$ in the next column $c'$, it discontinues the current traversal and transitions to the search phase.
Notably, in the latter case, it first fully traverses column $c'$ in direction $\N$ and afterwards moves to the first column west of $c'$.
This technical detail will play an important role in the runtime analysis.
While traversing a column in the $\S$ direction, the agent actively looks for an empty node that \emph{violates} the shape of a (partially filled) parallelogram with westernmost column $c$.
Specifically, it checks the two empty nodes immediately above (excluding column $c$) and below each column, as well as each empty neighbor to the east of the column.
Upon finding such a violating empty node $w$, the agent first places its carried tile at $w$ and then returns to column $c$ to retrieve the tile from $v$.
Subsequently, this exchange of tiles is termed as a \emph{tile shift} from $v$ to $w$ or as \emph{shifting} (the tile) from $v$ to $w$ (recall that the agent initially carries a tile that was never placed at any node).
After picking up the tile at $v$, the agent moves to an adjacent tile and transitions to the search phase again.
The agent terminates at the empty node $\S$ of the easternmost column once the configuration is fully traversed without encountering any violating nodes.
Any of the following conditions are sufficient for an empty node $w$ to be considered violating:
(1) $w$ has a tile at $\N$, $\NE$ and $\SE$ (e.g., \cref{subfig:parallelogram-c}), (2) $w$ has a tile at $\S$ and $\SE$ (e.g., \cref{subfig:parallelogram-d}) and is not $\N$ of the westernmost column (recall that we allow the parallelogram to be partially filled), (3) $w$ has a tile at $\NW$ and $\N$ (e.g., \cref{subfig:parallelogram-e,subfig:parallelogram-f,subfig:parallelogram-g,subfig:parallelogram-i}), (4) $w$ has a tile at $\NW, \SW$ and $\S$ (e.g., \cref{subfig:parallelogram-h}).

\subsection{An Icicle Formation Algorithm} \label{subsec:algorithm-icicle}

\newcommand{\tiledNeighb}{N_\tiled{}}

\begin{figure}[t]
    \centering
    \begin{subfigure}[c]{0.2\linewidth}
        \centering%
        \includegraphics[width=\linewidth,keepaspectratio,page=1]{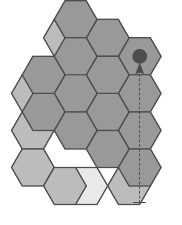}%
        \subcaption{}
        \label{subfig:project-A}
    \end{subfigure}%
    \hfill
    \begin{subfigure}[c]{0.2\linewidth}
        \centering%
        \includegraphics[width=\linewidth,keepaspectratio,page=2]{project}%
        \subcaption{}
        \label{subfig:project-B}
    \end{subfigure}%
    \hfill
    \begin{subfigure}[c]{0.2\linewidth}
        \centering%
        \includegraphics[width=\linewidth,keepaspectratio,page=3]{project}%
        \subcaption{}
        \label{subfig:project-C}
    \end{subfigure}%
    \hfill
    \begin{subfigure}[c]{0.2\linewidth}
        \centering%
        \includegraphics[width=\linewidth,keepaspectratio,page=4]{project}%
        \subcaption{}
        \label{subfig:project-D}
    \end{subfigure}%
    \hfill
    \begin{subfigure}[c]{0.2\linewidth}
        \centering%
        \includegraphics[width=\linewidth,keepaspectratio,page=5]{project}%
        \subcaption{}
        \label{subfig:project-E}
    \end{subfigure}%
    \linebreak
    \begin{subfigure}[c]{0.6\linewidth}
        \centering%
        \includegraphics[trim=0 0 0 -2.5mm,width=0.95\linewidth,keepaspectratio,page=6]{project}%
        \subcaption{}
        \label{subfig:project-heightOne}
    \end{subfigure}%
    \caption{During a projection, the agent (black disk) shifts each tile of a fragment in direction $\DSW$. Detailed in (a-d) is the projection of a single column; (e) is a snapshot of the configuration after the projection. The special case of a parallelogram with a height of one is shown in (f). To maintain connectivity in that case, the agent moves $\SW + \DNW$ to transition below the~next~column.}
    \label{fig:project}
\end{figure}

\begin{figure}[t]
    \begin{minipage}[b]{0.58\linewidth}%
        \begin{subfigure}[t]{0.165\linewidth}%
            \centering\includegraphics[width=0.95\linewidth,keepaspectratio,page=1]{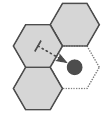}\subcaption{}\label{subfig:criticalCases-A}%
        \end{subfigure}%
        \begin{subfigure}[t]{0.165\linewidth}%
            \centering\includegraphics[width=0.95\linewidth,keepaspectratio,page=2]{criticalCases}\subcaption{}\label{subfig:criticalCases-B}%
        \end{subfigure}%
        \begin{subfigure}[t]{0.165\linewidth}%
            \centering\includegraphics[width=0.95\linewidth,keepaspectratio,page=3]{criticalCases}\subcaption{}\label{subfig:criticalCases-C}%
        \end{subfigure}%
        \begin{subfigure}[t]{0.165\linewidth}%
            \centering\includegraphics[width=0.95\linewidth,keepaspectratio,page=4]{criticalCases}\subcaption{}\label{subfig:criticalCases-D}%
        \end{subfigure}%
        \begin{subfigure}[t]{0.165\linewidth}%
            \centering\includegraphics[width=0.95\linewidth,keepaspectratio,page=5]{criticalCases}\subcaption{}\label{subfig:criticalCases-E}%
        \end{subfigure}%
        \begin{subfigure}[t]{0.165\linewidth}%
            \centering\includegraphics[width=0.95\linewidth,keepaspectratio,page=16]{criticalCases}\subcaption{}\label{subfig:criticalCases-F}%
        \end{subfigure}\linebreak%
        \begin{subfigure}[t]{0.165\linewidth}%
            \centering\includegraphics[width=0.95\linewidth,keepaspectratio,page=17]{criticalCases}\subcaption{}\label{subfig:criticalCases-G}%
        \end{subfigure}%
        \begin{subfigure}[t]{0.165\linewidth}%
            \centering\includegraphics[width=0.95\linewidth,keepaspectratio,page=6]{criticalCases}\subcaption{}\label{subfig:criticalCases-H}%
        \end{subfigure}%
        \begin{subfigure}[t]{0.165\linewidth}%
            \centering\includegraphics[width=0.95\linewidth,keepaspectratio,page=7]{criticalCases}\subcaption{}\label{subfig:criticalCases-I}%
        \end{subfigure}%
        \begin{subfigure}[t]{0.165\linewidth}%
            \centering\includegraphics[width=0.95\linewidth,keepaspectratio,page=8]{criticalCases}\subcaption{}\label{subfig:criticalCases-J}%
        \end{subfigure}%
        \begin{subfigure}[t]{0.165\linewidth}%
            \centering\includegraphics[width=0.95\linewidth,keepaspectratio,page=9]{criticalCases}\subcaption{}\label{subfig:criticalCases-K}%
        \end{subfigure}%
        \begin{subfigure}[t]{0.165\linewidth}%
            \centering\includegraphics[width=0.95\linewidth,keepaspectratio,page=10]{criticalCases}\subcaption{}\label{subfig:criticalCases-L}%
        \end{subfigure}\linebreak%
        \begin{subfigure}[t]{0.165\linewidth}%
            \centering\includegraphics[width=0.95\linewidth,keepaspectratio,page=11]{criticalCases}\subcaption{}\label{subfig:criticalCases-M}%
        \end{subfigure}%
        \begin{subfigure}[t]{0.165\linewidth}%
            \centering\includegraphics[width=0.95\linewidth,keepaspectratio,page=12]{criticalCases}\subcaption{}\label{subfig:criticalCases-N}%
        \end{subfigure}%
        \begin{subfigure}[t]{0.165\linewidth}%
            \centering\includegraphics[width=0.95\linewidth,keepaspectratio,page=13]{criticalCases}\subcaption{}\label{subfig:criticalCases-O}%
        \end{subfigure}%
        \begin{subfigure}[t]{0.165\linewidth}%
            \centering\includegraphics[width=0.95\linewidth,keepaspectratio,page=14]{criticalCases}\subcaption{}\label{subfig:criticalCases-P}%
        \end{subfigure}%
        \begin{subfigure}[t]{0.165\linewidth}%
            \centering\includegraphics[width=0.95\linewidth,keepaspectratio,page=15]{criticalCases}\subcaption{}\label{subfig:criticalCases-Q}%
        \end{subfigure}%
        \begin{subfigure}[t]{0.165\linewidth}%
            \centering\includegraphics[width=0.95\linewidth,keepaspectratio,page=18]{criticalCases}\subcaption{}\label{subfig:criticalCases-R}%
        \end{subfigure}%
    \end{minipage}%
    \begin{minipage}[b]{0.42\linewidth}%
        \begin{subfigure}[t]{0.5\linewidth}%
            \centering\includegraphics[trim=0 -9.5mm 0 0,width=0.95\linewidth,keepaspectratio,page=20]{criticalCases}\subcaption{}\label{subfig:criticalCases-S}%
        \end{subfigure}%
        \begin{subfigure}[t]{0.5\linewidth}%
            \centering\includegraphics[trim=0 -9.5mm 0 0,width=0.95\linewidth,keepaspectratio,page=19]{criticalCases}\subcaption{}\label{subfig:criticalCases-T}%
        \end{subfigure}%
    \end{minipage}%
    \caption{Illustrating all scenarios in which the northernmost node of a locally westernmost column is not removable. For brevity, (a–-r) only illustrate the agent's movement (indicated by arrows) to the empty node (with a dashed outline) that is tiled next; (s) and (t) also portray the subsequent tile shifts. In (r), the agent may alternatively enter \buildPar{} if the outlined node were tiled. Note that (s) and (t) only show instances where a tile at $\DSW$ (s) and $\DE$ (t) is encountered.}
    \label{fig:criticalCases}
\end{figure}
From a high-level perspective, the construction of an icicle involves the iterative transformation of a locally uppermost fragment into a parallelogram, followed by a projection of the fragment in the $\DSW$ direction.
When applying the parallelogram construction algorithm in a 2D configuration, the agent can always shift the tile at the northernmost node $v$ of a locally westernmost column without violating connectivity (\cref{subfig:criticalCases-A} illustrates that connectivity is preserved in the only critical 2D case).
In a 3D configuration, the situation becomes more intricate.
There are multiple cases in which the tile at $v$ must remain in its immediate neighborhood to avoid violating connectivity.
Additionally, there is a case in which the tile at $v$ cannot be moved at all unless neighboring tiles are also moved.
We categorize these cases based on specific properties of node $v$, which we define as follows:

\begin{definition}\label{def:removableShiftableUnmovable}
    Let $v \in \tiled{}$ be an arbitrary tiled node. Denote by $N(v)$ the neighborhood of $v$ (exluding $v$), and by $\tiledNeighb{}(v)$ its subset of only tiled nodes.
    Node $v$ is \emph{removable}, if the tiled neighbors of $v$ are locally connected, i.e., $G|_{\tiledNeighb{}(v)}$ is connected.
    Node $v$ is \emph{shiftable}, if $G|_{\tiledNeighb{}(v)}$ is disconnected and there exists a node $w \neq v$ (termed \emph{bridge node} of $v$) for which $G|_{\tiledNeighb{}(v)\cup \{w\}}$ is connected.
    Any node that is neither removable nor shiftable is termed~\emph{unmovable}.
\end{definition}

We now state the full icicle algorithm:
The agent starts in the search phase where it repeatedly moves $\UW, \USE, \UNE, \NW, \SW$ and $\N$ until it eventually stops at some node $v$.

If node $v$ is removable, the parallelogram traversal procedure \buildPar{} is entered. 
There are three possible outcomes: the agent returns from the procedure after finding a more western column or some tile above, after placing a tile, or at the empty node $\S$ of the fragment's easternmost column.
In the first case, the agent transitions to the search phase. 
In the second case, the agent first moves back to pick up the tile at node $v$, then moves to the next tile at $\S$ or $\SE$, and afterwards transitions to the search phase.
In the third case, the current fragment forms a correctly shaped parallelogram and the agent proceeds by executing the \proj{} procedure.
During \proj{}, each tile of the fragment is projected in the $\DSW$ direction.
Starting with the easternmost column, tiles are projected columnwise from east to west and within the columns from $\N$ to $\S$ (see \cref{subfig:project-A,subfig:project-B,subfig:project-C,subfig:project-D,subfig:project-E}).
Let $v_0,...,v_k$ be the nodes of the currently projected column ordered from $\N$ to $\S$.
For each $i = 0,...,k$, the agent performs a tile shift from $v_i$ to the first empty node $w_i$ in direction $\DSW$ of $v_i$.
After picking up the last tile of the column at $v_k$, the agent moves $\NW$ and continues the projection in the western neighboring column.
In the special case of a degenerated parallelogram with a height of one, after picking up a tile, the agent moves $\SW$ and $\DNW$ instead (see \cref{subfig:project-heightOne}).
These additional steps ensure that connectivity is maintained during the projection.
Once the last tile of the fragment is projected, the agent transitions to the search phase in the layer below.

Otherwise, if the agent stops at a non-removable node $v$, it acts according to the case distinction outlined below, prioritized in the given order (refer to \cref{fig:criticalCases} for a graphical overview).
Subsequently, the agent transitions to the search phase, concluding our algorithm.

\begin{case} \label{case:a}
    If $v + \SE$ is a bridge node of $v$ (thereby $v$ is shiftable and $v + \SE$ is empty), then shift the tile from $v$ to $v + \SE$ (see \cref{subfig:criticalCases-A,subfig:criticalCases-B,subfig:criticalCases-C,subfig:criticalCases-D,subfig:criticalCases-E,subfig:criticalCases-F,subfig:criticalCases-G}), and move to $v + \S$ afterwards.
\end{case}

\begin{case} \label{case:b}
    If $v + \DE$ is a bridge node of $v$, and at least one neighboring tile is not at $\DSW$ or $\SE$, then shift the tile from $v$ to $v + \DE$ (see \cref{subfig:criticalCases-H,subfig:criticalCases-I,subfig:criticalCases-J,subfig:criticalCases-K,subfig:criticalCases-L,subfig:criticalCases-M,subfig:criticalCases-N,subfig:criticalCases-O,subfig:criticalCases-P,subfig:criticalCases-Q}), and move to the first tile at $v + \S$, $v + \SE$ or $v + \NE$.
    Additionally, if $v + \S$ is empty and both $v + \SE$ and $v + \NE$ are tiled, then traverse the next column starting at $v + \SE$ in direction $\S$. 
    If during that traversal a tile at $\UW, \USE$ or $\SW$ is encountered, then immediately transition to the search phase.
\end{case}

\begin{case} \label{case:c}
    If the only tiled neighbors of $v$ are at $\DSW$ and $\SE$, then move $\SE$ and observe node $w = v + \SE + \DSW$. If $w$ is empty, then shift the tile from $v$ to $w$ (see \cref{subfig:criticalCases-R}), and move to $v + \SE$ afterwards. 
    Otherwise, if $w$ is already tiled, move back to $v$ and enter \buildPar{}.
\end{case}

\begin{case} \label{case:d}
    If the only tiled neighbors of $v$ are at $\DNW, \S$ and $\NE$ ($v$ is unmovable, see \cref{subfig:criticalCases-S}), then follow these steps:
    First, move $\S$ until some node $w$ is entered that has a neighboring tile at $\UW, \USE, \SW, \DSW, \SE$ or $\DE$, or until there is no more tile in direction $\S$.
    If $w$ has a tile at $\UW, \USE$ or $\SW$, then immediately transition to the search phase.
    Otherwise, shift each tile in the column that is somewhere $\N$ of $w$ in direction $\DE$ (including $w$ if $w + \DE$ is empty).
    To be precise, let $v_k,v_{k-1},...,v_1, v$ be the nodes of the column ordered from $\S$ to $\N$ starting at $v_k = w + \N$ (or $v_k = w$ if $w + \DE$ is empty).
    Perform a tile shift from $v_i$ to $v_i + \DE$ for each $i$ with $k \geq i > 0$. 
    After the tile shift at $v_i$ with $i > 0$, move $\NE + \DNW$ to be positioned at $v_{i-1} + \DE$ (to preserve connectivity).
    Once the final tile is picked up, move to $v + \NE$.
\end{case}

\begin{case} \label{case:e}
    If the only tiled neighbors of $v$ are at $\DNW$ and $\S$ (see \cref{subfig:criticalCases-T}), then proceed analogously to the previous case, with the exception that tiles at $\DSW$ are disregarded. 
    Additionally, make the following adaptations:
    If no tile at $\UW, \USE, \SW, \SE$ or $\DE$ is encountered, then project the whole column (which is a parallelogram of width one) in direction $\DSW$.
    Otherwise, after performing the final tile shift in direction $\DE$, repeatedly move $\S$ (on empty nodes) and enter the first tiled node at $\S$ or $\SE$ (which must exist since we did not project).
\end{case}

The following remarks aim to clarify the choices made in the above case distinction:
In~\cref{case:c}, the node $v + \DE$ serves as a bridge node for $v$; however, the agent takes an additional step by attempting to shift the tile to $v + \SE + \DSW$.
This decision stems from the fact that $v + \DE$ is not within the bounding cylinder of tiles observable from node $v$.
Similarly, in \cref{case:d}, $v + \DSW$ serves as a bridge node for $v$.
Although $v + \DSW$ is within the bounding cylinder of observable tiles, it shares the same $x$- and $y$-coordinates as $v$. 
It is essential to our analysis that tiles are never placed outside of the bounding cylinder, and that, except for projections, tiles consistently advance to the east or south.
In \cref{case:b}, the agent removes the last tile of some column $c$ and instead of immediately transitioning to the search phase, it first traverses the next column $c'$ in the $\S$ direction.
Similarly to the \buildPar{} procedure, where the agent first traverses the next column $c'$ fully in direction $\N$ whenever a more northern column of the same $x$-coordinate as $c$ is found, this additional traversal is crucial for the runtime analysis.
To elaborate, if column $c'$ has multiple adjacent columns to the west, then directly entering the search phase would result in repeatedly traversing the same tiles within column $c'$.
However, with the additional traversal in direction $\S$ in \cref{case:b} (and in direction $\N$ in procedure \buildPar{}), we can ensure that each tile of $c'$ is visited only a constant number of times whenever the agent does not currently perform a tile shift.

\section{Analysis} \label{sec:analysis}

\newcommand{\propertyOne}{\hyperref[def:fragmentProperties]{P1}}
\newcommand{\propertyTwo}{\hyperref[def:fragmentProperties]{P2}}
\newcommand{\propertyThree}{\hyperref[def:fragmentProperties]{P3}}
\newcommand{\properties}[1]{\hyperref[def:fragmentProperties]{#1}}
\newcommand{\allProperties}{\hyperref[def:fragmentProperties]{P1--P3}}

\newcommand{\platforms}{\mathcal{P}}
\newcommand{\allCases}{\cref{case:a,case:b,case:c,case:d,case:e}}

Our proof structure is as follows.
We first show that our algorithm complies with the connectivity constraint of the 3D hybrid model.
Following this, we introduce three key definitions, termed \allProperties{}, each associated with some fragment of a configuration.
Moving forward, we show convergence towards an icicle in two steps.
First, we show that any initially connected configuration converges to a configuration containing a fragment that satisfies \allProperties{}.
Second, we prove that the aforementioned configuration further converges to an icicle.
In the following, $C^i = (\tiled{}^i, p^i)$ denotes the configuration that results from the execution of \icicleAlgo{} for $i$ steps.

\begin{lemma}\label{lem:connectivityIsMaintained}
    If the agent disconnects $G|_{\tiled{}^i}$ in step $i$, then $G|_{\tiled{}^{i+4}}$ is connected, and for all $i < j < i+4$: $G|_{\tiled{}^{j} \cup \{p^j\}}$ is connected and the agent carries a tile.
\end{lemma}

\begin{proof}
    To disconnect $G|_{\tiled{}^i}$, the agent must pick up a tile from some node $v$ in step $i$.
    Note that the lemma's statement holds trivially for step $j = i+1$, since the agent's position does not change by picking up a tile.
    Node $v$ must be a cut node w.r.t. $G|_{\tiled{}^i}$, which implies that it is also a cut node w.r.t. $G|_{\tiledNeighb{}(v)}$, i.e., not a removable node.
    Hence, we can exclude procedure \buildPar{}, as it is only entered if node $v$ were removable.
    In \cref{case:a,case:b,case:c}, a bridge node is tiled first, which by \cref{def:removableShiftableUnmovable} maintains connectivity between nodes in $\tiledNeighb{}(v)$.
    Hence, we must only consider the \proj{} procedure, and \cref{case:d,case:e}.

    Consider the \proj{} procedure.
    After a projection at node $v$, the node $v + \DSW$ must be tiled.
    Node $v + \DSW$ is adjacent to $v + \SW$, a node whose $y$-coordinate is by one smaller, i.e., $y(v + \SW) = y(v) - 1$.
    This implies that connectivity is maintained if the parallelogram's height is larger than one.
    Note that since tiles are projected from $\N$ to $\S$, this also ensures connectivity for the potentially partially filled westernmost column.
    Consider the projection of a parallelogram of height one.
    After each tile shift the agent positions itself below the next column by moving $\SW$ and $\DNW$.
    Both $v + \SW$ and $v + \SW + \DNW$ are adjacent to $v + \DSW$ (which must be tiled) and $v + \NW$ (which is the only node in the next column).
    Thereby, $G|_{\tiled{}^{i+2} \cup \{p^{i+2}\}}$ and $G|_{\tiled{}^{i+3} \cup \{p^{i+3}\}}$ are connected, and after placing the tile at $v + \SW + \DNW$ in step $i+3$, $G|_{\tiled{}^{i+4}}$ is connected again.

    During \cref{case:d,case:e}, the agent performs tile shifts in direction $\DE$, i.e., after picking up a tile at node $v$, $v + \DE$ must be tiled.
    Afterwards, the agent positions itself below the tile at $\N$ by moving $\NE$ and $\DNW$.
    Similarly, $v + \NE$ and $v + \NE + \DNW$ are adjacent to $v + \DE$ (which we know is tiled) and $v + \N$ (which is shifted next) and the proof is analogous.
\end{proof}

With the following auxiliary lemma, we show that tiles are never placed outside of the bounding cylinder $\bc{}(\tiled{})$ .

\begin{lemma}\label{lem:maintainBoundingCylinder}
    If during the execution of \icicleAlgo{} a tile is shifted from some node $v$ to some node $w$, then there are tiled nodes $u_x,u_y \in \tiled{}$ with $x(w) = x(u_x)$ and $y(w) = y(u_y)$.
\end{lemma}

\begin{proof}
    We prove the lemma by explicitly providing the nodes $u_x, u_y$ that have the same $x$- and $y$-coordinate as node $w$, respectively.
    Note that $u_x = v$ or $u_y = v$ can be a valid choice, since the tile is placed at $w$ before the tile at $v$ is picked up.
    According to \cref{eq:transition}, nodes in direction $\N$ and $\S$ have equal $x$-coordinate, nodes in direction $\NW$ and $\SE$ have equal $y$-coordinate, and in direction $\UNE$ and $\DSW$ have both equal $x$-and $y$-coordinates.

    If the tile shift is performed during the \proj{} procedure, then $u_x = u_y = w + \UNE$ must be tiled and our claim follows.
    If it is performed during the \buildPar{} procedure, then there are four cases: (1) $w$ has tiles at $\N, \NE$ and $\SE$, (2) $w$ has tiles at $\S$ and $\SE$, (3) $w$ has tiles at $\NW$ and $\N$, or (4) $w$ has tiles at $\NW, \SW$ and $\S$.
    In each case we can choose a tiled node at $\N$ or $\S$ for $u_x$, and a tiled node at $\NW$ or $\SW$ for $u_y$ such that our claim follows.

    Any tile shift outside of these procedures is performed in one the \allCases{}.
    Consider the tile shift where node $v$ is the northernmost node of the column considered in these cases.
    Node $v$ cannot have a tiled neighbor in direction $\UW, \USE, \UNE, \NW, \SW$ or $\N$, as otherwise the agent would not leave the search phase at $v$.
    This implies that the candidates for tiled neighbors of $v$ are in directions $\DE, \DNW, \DSW, \SE, \NE$ and $\S$.
    Hence, there are $2^6 = 64$ possible neighborhoods.
    The $20$ cases in which $v$ is not removable are depicted in \cref{fig:criticalCases}.
    For completeness, \cref{fig:nonCriticalCases} depicts the remaining $44$ cases.
    One can easily verify that $v$ is removable in each of those cases, i.e., our case distinction is complete.

    In the following we provide the nodes $u_x$ and $u_y$ that satisfy our claim for each of the $20$ non-removable cases (enumerated according to \cref{fig:criticalCases}).
    In cases (a)--(e), $u_x = w + \N$ and $u_y = v$.
    In cases (f) and (g),  $u_x = w + \DNW$ and $u_y = v$.
    In cases (h)--(i), (k)--(m) and (p)--(r), $u_x = u_y = w + \UNE$.
    In cases (j) and (n)--(o), $u_x = w + \USE$ and $u_y = w + \NW$.
    Finally, in cases (s) and (t), multiple tile shifts are performed precisely from the nodes $v_k,...,v_0$ to the nodes $w_k = v_k + \DE, ..., w_0 = v_0 + \DE$ in the given order, where $v_k$ is the $\S$-most node at which a tile shift is performed, and $v_0 = v$ the $\N$-most node of the column.
    For each $i > 0$, $u_y = v_{i-1}$ is a valid choice, and for $i = 0$, $u_y = v_0 + \NW$.
    In case (s), we can simply choose $u_x = v + \NE$ for all $i$.
    In case (t), the node $v + \NE$ is empty.
    Here, either $v_k + \S + \DE$ or $v_k + \SE$ is a valid choice for $u_x$.
    Note that either of the two nodes must be tiled, as otherwise no tile shift at $v_k$ would be performed in case (t).

    For all possible tile shifts that can be performed during the execution of \icicleAlgo{} we explicitly provided the nodes $u_x$ and $u_y$, which concludes the lemma.
\end{proof}

\begin{figure}[t]
    \centering
    \includegraphics[width=\linewidth,keepaspectratio,page=1]{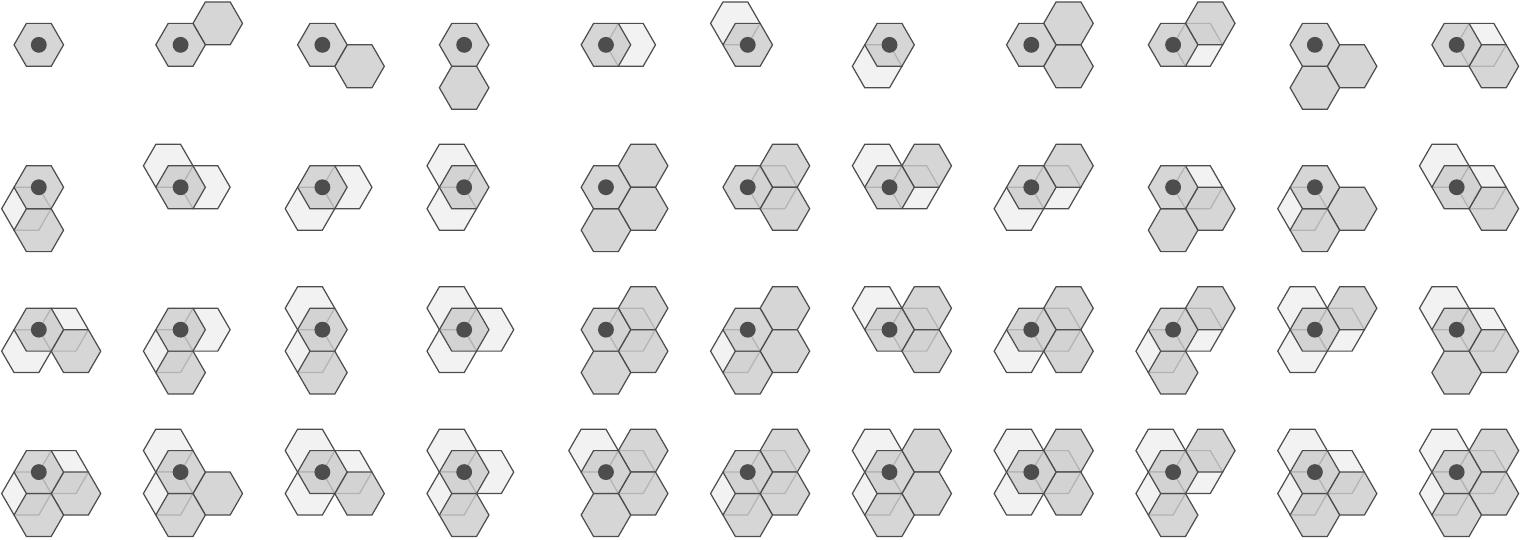}%
    \caption{Illustrating all 44 possible neighborhoods $\tiledNeighb{}(v)$ for the case where the agent (black disk) leaves the moving phase at a removable node $v$. In these cases, the agent enters the \buildPar{} procedure, i.e., they do not need to be handled explicitly by the icicle-formation algorithm. Refer to \cref{fig:criticalCases} for the 20 neighborhoods that are handled explicitly by \cref{case:a,case:b,case:c,case:d,case:e}, and note that there are precisely $2^6 = 64$ possible neighborhoods, i.e., our listing is complete.}
    \label{fig:nonCriticalCases}
\end{figure}

\begin{lemma}\label{lem:coordinateZeroNodesExist}
    For each $i\geq 0$ there is a tiled node $v \in \tiled{}^i$ with $x(v) = 0$.
\end{lemma}

\begin{proof}
    Assume by contradiction, that the lemma's statement does not hold, and let $i$ be the first step for which there is no tiled node $v \in \tiled{}^i$ with $x(v) = 0$.
    Recall that the $x$-coordinate of any node is relative to the initial minimum $x$-coordinate in $\tiled^0$, which we defined as $x_{min} = 0$.
    Hence, it holds that $i > 0$ and in step $i-1$ the agent has picked up a tile at the only node $v \in \tiled{}^{i-1}$ with $x(v) = 0$.
    We again distinguish whether the tile was picked up as part of a tile shift in the procedures \proj{}, \buildPar{} or in \allCases{}, and lead each case to a contradiction which concludes the lemma.

    First, in procedure \proj{}, between placing a tile at some node $w$ and picking up a tile from node $v$, the agent moves exclusively $\UNE$.
    \cref{eq:transition} implies that $x(w) = x(v)$, which contradicts that $v$ is the only node with $x(v) = 0$.

    Second, consider a tile shift from $v$ to some node $w$ that was initiated in procedure \buildPar{}.
    There are two cases:
    Either $w$ is the empty node $\S$ of the column $c$ that contains $v$, or $w$ lies somewhere east of column $c$.
    The former case again contradicts that $v$ is the only node with $x(v) = 0$.
    In the latter case, there must exist a node $u$ with $x(u) < 0$, which contradicts \cref{lem:maintainBoundingCylinder}.

    Third, consider \allCases{}, and recall that we have established in the proof of \cref{lem:maintainBoundingCylinder} that the 20 cases depicted in \cref{fig:criticalCases} are a complete list of all possible neighborhoods of node $v$.
    In each case there exists a tiled node $u$ at $\NE, \SE$ or $\DE$, which implies $x(u) = -1$ and thus contradicts \cref{lem:maintainBoundingCylinder}.
\end{proof}

We want to measure the progress of tiles within the bounding cylinder towards the east and south by considering their $x$- and $y$-coordinates.
As part of the \buildPar{} procedure and \cref{case:b,case:d,case:e}, the $y$-coordinate of tiles can increase when their $x$-coordinate decreases.
Although the size of the bounding cylinder cannot increase by \cref{lem:maintainBoundingCylinder}, it may decrease.
In such instances, by \cref{lem:coordinateZeroNodesExist}, the resulting bounding cylinder always aligns with the eastern side of the initial bounding cylinder $\bc(\tiled{}^0)$.
To address this, we introduce a combined representation of the $x$- and $y$-coordinates w.r.t. the bounding cylinder $\bc(\tiled{})$ for arbitrary $\tiled{}$.

Let $y_{max}^\tiled{}$ and $y_{min}^\tiled{}$ be the maximum and minimum $y$-coordinates within $\bc(\tiled{})$, and let $h = y_{max}^\tiled{} - y_{min}^\tiled{} + 1$ be the \emph{height} of $\bc(\tiled{})$, i.e., the cylinder's extent along the $y$-axis.
We define the $xy$\emph{-coordinate} of some node $v \in \bc(\tiled{})$ as $xy(v) = x(v) \cdot h + y(v) - y_{min}^\tiled{}$.

Consider the following definitions, which we refer to as \allProperties{}, that relate to some fragment $F \subseteq \tiled{}$.
\propertyOne{} characterizes a locally uppermost fragment, \propertyTwo{} a fragment covering the $xy$-coordinates of all tiled nodes, and \propertyThree{} a fragment wherein tiles have the shape of a parallelogram aligned along the southern, eastern, and northern sides of the bounding~cylinder.

\begin{definition}\label{def:fragmentProperties}
    Let $F \subset \tiled{}$ be an arbitrary fragment.
    \begin{itemize}
        \item \emph{P1:} $F$ is a \emph{platform}, if $\{v + \X \mid v \in F, \X \in \{\UW, \USE, \UNE\}\} \cap \tiled{}$ is an empty set. 
        \item \emph{P2:} $F$ is \emph{covering}, if for each node $v \in \tiled{}$ there is a node $w \in F$ with $xy(w) = xy(v)$.
        \item \emph{P3:} $F$ is an \emph{aligned parallelogram}, if for each node $v \in F$ it holds that for all $i$ with $xy(v) \geq i \geq 0$ there is a node $w \in F$ with $xy(w) = i$.
    \end{itemize}
\end{definition}

We can now use \allProperties{} to give an alternative definition of the icicle shape. 

\begin{definition}\label{def:icicle}
    A Configuration $C = (\tiled{},p)$ is an \emph{icicle}, if it contains a fragment $F$ that satisfies \allProperties{}, and for any node $v \in \tiled{} \setminus F$ it holds that $v + \UNE \in \tiled{}$.
\end{definition}

Any tiled node that is not contained in the fragment $F$ specified in \cref{def:icicle}, must be somewhere $\DSW = - \UNE$ of $F$, as otherwise the number of tiles would be infinite.
Hence, each node $v \in \tiled{} \setminus F$ is contained in a tower of tiles whose uppermost tile is contained in $F$, and thereby \cref{def:icicle} is equivalent to our definition of an icicle from \cref{subsec:problem}.

Subsequently, we only consider configurations in which the agent leaves the search phase at some node $v$. 
This must eventually occur since moving upwards increases its $z$-coordinate, and moving in directions $\SW$, $\NW$, or $\N$ increases its $xy$-coordinate. 
Both coordinates are bounded within any finite set of tiled nodes.
To simplify notation, we use $C^i = (\tiled^i, p^i)$ to represent the configuration where the agent leaves the search phase for the $i$-th time.

Consider the potential function $\Phi^i = \sum_{v \in \tiled^i}xy(v) + |\platforms^i|$, where $\platforms^i$ denotes the set of all platforms, i.e., fragments satisfying \propertyOne{}. 
We first show its monotonicity and lower bound.

\begin{lemma}\label{lem:potential1Monotonicity}
    For each $i \geq 0$ it holds that $\Phi^i \geq \Phi^{i+1} \geq 0$, and if $\Phi^i = \Phi^{i+1}$, then (1) no tile was shifted between step $i$ and $i+1$, or (2) a fragment was projected between step $i$ and $i+1$.
\end{lemma}

\begin{proof}
    By definition, $xy(v) \geq 0$ holds for all initially tiled nodes $v \in \tiled{}^0$.
    From \cref{lem:maintainBoundingCylinder} follows that $xy(v) \geq 0$ holds for all $v \in \tiled{}^i$ for all $i \geq 0$.
    Since the number $|\platforms^i|$ of platforms cannot be negative, $\Phi^i \geq 0$ follows for all $i \geq 0$, which concludes the potential's lower bound.

    We proceed by showing monotonicity.
    If $\tiled^i = \tiled^{i+1}$ holds, then $\Phi^i = \Phi^{i+1}$ follows trivially.
    Otherwise, at least one tile was shifted before leaving the search phase for the $i+1$-th time.
    We distinguish whether a tile was shifted in \proj{}, \buildPar{} or in \allCases{}.

    First, consider the projection of a fragment $F$ in procedure \proj{}.
    For each $v \in F$, a tile shift from $v$ to the first empty node $w$ somewhere in direction $\DSW$ is performed.
    The number of platforms increases at most by one, since all nodes in $\{v + \DSW \mid v \in F\}$ are contained in the same fragment in step $i+1$, which may or may not be a platform.
    At the same time, the number of platforms decreases precisely by one, since all nodes in $F$ are empty in step $i+1$, and $F$ must be a platform in step $i$, since only platforms are projected.
    Thereby, $|\platforms^i| \geq |\platforms^{i+1}|$ holds.
    \cref{eq:transition} directly implies $\sum_{v \in \tiled^i}xy(v) = \sum_{v \in \tiled^{i+1}}xy(v)$.
    This concludes  $\Phi^i \geq \Phi^{i+1}$ for the \proj{} procedure.

    Second, if a tile is shifted from node $v$ to some node $w$ in procedure \buildPar{}, then $v$ must be removable.
    This implies that the number of platforms cannot increase by picking up the tile from $v$.
    It cannot increase by placing a tile at $w$ either, since $w$ must be adjacent to some node of the fragment containing $v$.
    Following the columnwise traversal, $w$ is either the node $\S$ of the column containing $v$, or it is a node located east of $v$.
    In both cases, $xy(v) > xy(w)$ holds, concluding strict monotonicity $\Phi^i > \Phi^{i+1}$ for the \buildPar{} procedure.

    Third, consider \allCases.
    A tile is shifted from node $v$ to $v + \SE$ (\cref{case:a}), to $v + \DE$ (\cref{case:b}), or to $v + \SE + \DSW$ (\cref{case:c}).
    It holds that $xy(v + \SE) = xy(v) - h$, $xy(v + \DE) = xy(v) - h + 1$ and $xy(v + \SE + \DSW) = xy(v) - h$.
    Additionally, $h > 1$ holds in \allCases, which is easy to see by examining all possible neighborhoods of node $v$ provided in \cref{fig:criticalCases}.
    We further observe that \cref{case:d} is the only case in which the number of platforms can increase, and if it does, it does so by precisely one.
    However, at the same time at least two tiles are shifted in direction $\DE$, which decreases the potential by at least $2h - 2 > 2$.
    Finally, in \cref{case:e}, the agent either enters the \proj{} procedure, as we have already analyzed in this proof, or tiles are shifted exclusively in the $\DE$ direction, similar to \cref{case:b}.
    Therefore, in each case, we have $\sum_{v \in \tiled^i}xy(v) > \sum_{v \in \tiled^{i+1}}xy(v)$.
    Moreover, in the only case where $|\platforms^i| < |\platforms^{i+1}|$ holds (\cref{case:d}), it follows that $|\platforms^{i+1}| - |\platforms^i| < \sum_{v \in \tiled^i}xy(v) - \sum_{v \in \tiled^{i+1}}xy(v)$, i.e., the potential's increase is compensated by its decrease such that $\Phi^i > \Phi^{i+1}$ holds in \allCases{}.

    Except for the \proj{} procedure, during which $\Phi$ is non-increasing, we have shown that $\Phi$ is strictly decreasing, and lower bounded by zero which concludes the lemma.
\end{proof}

In the following lemma, we show that once the agent is positioned within a fragment that satisfies \allProperties{}, it will remain in such a fragment.
Subsequently, we use the potential $\Phi$ and its monotonicity to prove that this configuration is eventually reached.

\begin{lemma}\label{lem:propertiesMaintained}
If $p^i \in F^i$ where $F^i$ is a fragment in $C^i$ that satisfies \allProperties{}, then $p^{i+1} \in F^{i+1}$ where $F^{i+1}$ is a fragment in $C^{i+1}$ that satisfies \allProperties{}.
\end{lemma}

\begin{proof}
    As we will show in the next paragraph, all tiles at nodes in $F^i$ must have been projected between step $i$ and $i+1$ such that each node directly $\DSW$ of any node in $F^i$ is tiled in step $i+1$.
    Let $F$ be the fragment that contains all nodes directly $\DSW$ of $F^i$ in step $i+1$, i.e.,  $F \supseteq \{v \mid v + \UNE \in F^i\}$.
    Since $F^i$ satisfies \propertyTwo{}, there does not exist a node in $F$ with an $xy$-coordinate that was not contained in $F^i$, i.e., $F = \{v \mid v + \UNE \in F^i\}$.
    It follows that $F$ satisfies \propertyTwo{} and \propertyThree{} as well.
    Since in step $i+1$, each node in $F^i$ is empty, $F$ must be a platform (\propertyOne{}).
    This implies that after the projection of $F^i$, the agent cannot leave $F$ to any layer above.
    It follows that $F^{i+1} = F$, which concludes the lemma.

    Assume that $F^i$ has height at least two.
    In the special case where $F^i$ is a parallelogram of height one (refer to \cref{subfig:criticalCases-R}), the agent can enter \cref{case:c} instead of the \buildPar{} procedure at node $v$, and subsequently either shift the tile from $v$ to $v + \SE + \DSW$ or proceed with the projection of $F^i$.
    Note that since $F^i$ satisfies \propertyTwo{}, each fragment below $F^i$ must also have height one, which implies (1) that the configuration is already an icicle, and (2) that no tile is ever shifted outside of \cref{case:c} and the \proj{} procedure, such that the proof is analogous with slightly different notation.
    Hence, the assumption that $F^i$ has height at least two does not lose generality.
    The projection of $F^i$ between step $i$ and $i+1$ can be deduced as follows:
    The agent cannot leave $F^i$ to a layer above since $F^i$ is a platform (\propertyOne{}).
    Since $F^i$ is a covering fragment (\propertyTwo{}), and all \allCases{} necessitate a neighboring tile below that is not covered by $F^i$, it follows that the northernmost node $v$ of its westernmost column is removable and the \buildPar{} procedure must be entered.
    Lastly, since $F^i$ is an aligned parallelogram (\propertyThree{}), in the \buildPar{} procedure, where the agent visits all nodes of $F^i$ in descending order of their $xy$-coordinates, it cannot find a violating node.
    Consequently, the agent must proceed to the \proj{} procedure.
\end{proof}

\begin{lemma}\label{lem:convergenceIntermediate}
    For each $i \geq 0$ there is a step $j > i$ such that (1) $\Phi^i > \Phi^j$ or (2) $p^j \in F^j$ where $F^j$ is a fragment of configuration $C^j$ that satisfies \allProperties{}.
\end{lemma}

\begin{proof}
    Let $F^i$ be the fragment in which the agent leaves phase search in step $i$, i.e., $p^i \in F^i$.
    If $F^i$ satisfies \allProperties{}, then the lemma follows directly from \cref{lem:propertiesMaintained}.
    Assume that $F^i$ does not satisfy \allProperties{}. 
    If $\Phi^i = \Phi^{i+1}$, then by \cref{lem:potential1Monotonicity} either no tile was shifted or fragment $F^i$ was projected between step $i$ and $i+1$.
    The agent does not shift a tile, whenever it finds a more western column or some fragment above, in which case $z(p^{i+1}) > z(p^i)$ or $xy(p^{i+1}) > xy(p^i)$ holds.
    Given that the agent's $z$-coordinate can increase by at most $n$, and its $xy$-coordinate by at most $n^2$, some tile must eventually be shifted.
    If the tile shift takes place outside of the \proj{} procedure, then there is nothing to prove, as the lemma directly follows from \cref{lem:potential1Monotonicity} in such cases.
    Therefore, for the remaining proof, we assume that each tile shift is executed as part of the \proj{} procedure.
    For the sake of clarity in notation, we assume that $F^i$ is projected directly between step $i$ and $i+1$.

    The projection of $F^i$ can only be initiated after an execution of the \buildPar{} procedure.
    Within this procedure, the agent visits all nodes of $F^i$ in descending order of their $xy$-coordinates, and it does not find any empty node in the bounding box of $F^i$ with an $xy$-coordinate smaller than any node of $F^i$.
    This implies that $F^i$ is locally an aligned parallelogram, satisfying \propertyThree{} with respect to configuration $C = (F^i,p^i)$.
    The agent also does not find a neighboring tile above, which implies that $F^i$ satisfies \propertyOne{}.
    Note that $F^i$ does not necessarily satisfy \propertyThree{} globally, i.e., w.r.t. $C^i = (\tiled{}^i,p^i)$.
    Consequently, $F^i$ cannot satisfy \propertyTwo{}, as any parallelogram that satisfies \propertyTwo{} must also satisfy \propertyThree{}.

    Since $F^i$ does not satisfy \propertyTwo{}, it follows that there exists a node $w \in \tiled^{i}$ such that $xy(w) \neq xy(v)$ for all $v \in F^i$.
    Given that connectivity is maintained by \cref{lem:connectivityIsMaintained}, there is a path $P$ from any node $v \in F^i$ to $w$ in $G|_{\tiled{}^i}$.
    As $F^i$ is a platform, the first node $u$ on that path, for which $xy(u)$ is distinct from any $xy$-coordinate in $F^i$, must have a smaller $z$-coordinate than any node in $F^i$.
    This implies that after precisely $k = |z(u) - z(v)|$ projections, the agent must be positioned in a larger fragment, i.e., $|F^{i+k}| > |F^i|$.

    If, in step $i+k$, $F^{i+k}$ is not a platform, then the series of projections carried out between steps $i$ and $i+k$ results in a reduction in the number of platforms, in which case the lemma follows directly.
    Otherwise, if $F^{i+k}$ is not a parallelogram in step $i+k$, the next tile shift must be executed outside the \proj{} procedure, and the lemma follows from \cref{lem:potential1Monotonicity}.
    Finally, if $F^{i+k}$ is both a platform and a parallelogram but does not satisfy \propertyTwo{}, we can once again identify a node $w'$ whose $xy$-coordinate is not covered by any node of $F^{i+k}$, and the process described above repeats for $w'$.
    The size of any fragment is upper bounded by $n$, which implies that this process cannot repeat indefinitely.
    Note that distinct nodes of the same layer must have distinct $xy$-coordinates, which implies that whenever the size of the fragment at the agent's position increases, the total number of $xy$-coordinates that are covered by that fragment also increases.
    Consequently, it follows that eventually, in some step $j$, one of the following conditions holds: (1) the number of platforms decreases, (2) the agent shifts a tile outside of the \proj{} procedure, or (3) $F^j$ satisfies \propertyTwo{}.

    Again, in the first two cases, $\Phi^i > \Phi^j$.
    In the latter case, whenever the first two cases do not occur, \properties{P1--P2} hold as well such that the lemma follows.
\end{proof}

The initial number of platforms is at most $n$, and the initial $xy$-coordinate of any tiled node is at most $n^2$. 
Hence, the initial potential is $\Phi^0 = \O(n^3)$.
Consequently, \cref{lem:propertiesMaintained} and \cref{lem:convergenceIntermediate} imply that eventually the agent is positioned in a fragment satisfying \allProperties{}.

For the second part of our analysis, dedicated to demonstrating convergence towards an icicle, we introduce another potential function $\Psi^i$. 
This function is defined as the number of empty nodes within the bounding box $\bb(\tiled{}^i)$ that have a tile somewhere in the $\DSW$ direction.
Formally, $U^i = \{v \in \bb{}(\tiled{}^i) \setminus \tiled^i \mid v + k \cdot \DSW \in \tiled{}^i \text{ for some } k > 0\}$ and $\Psi^i = |U^i|$.

\begin{lemma}\label{lem:potential2Monotonicity}
    Let $p^i \in F^i$ where $F^i$ satisfies \allProperties{}. If $\Psi^i > 0$, then $\Psi^{i+1} < \Psi^i$.
\end{lemma}

\begin{proof}
    Let $U^i \coloneqq \{v \in \bb{}(\tiled{}^i) \setminus \tiled^i \mid v + k \cdot \DSW \in \tiled{}^i \text{ for some } k > 0\}$ be the set from the definition of $\Psi^i$, and note that $U^i$ is not empty by assumption.
    We prove the lemma by showing that a node $u^i \in U^i$ with $u^i \notin U^{i+1}$ exists, and that $u^{i+1}\in U^i$ for all $u^{i+1}\in U^{i+1}$.

    Let $u^i \in U^i$ be a node with maximum $z$-coordinate.
    Recall that $F^i$ is a covering fragment as per \propertyTwo{} and also a platform as per \propertyOne{}.
    This impilies that no node within the bounding box $\bb(\tiled^{i})$ possesses a $z$-coordinate larger than $z(F^i)$.
    Here, $z(F^i)$ denotes the common $z$-coordinate of all nodes within $F^i$.
    Consequently, it follows that $z(u^i) \leq z(F^i)$.

    According to \cref{eq:transition}, consecutive nodes in both the $\DSW$ and $\UNE$ directions share identical $xy$-coordinates, i.e., for all $w = v + k \cdot \DSW$ it holds that $xy(w) = xy(v)$.
    For all empty nodes $v$ of $z$-coordinate $z(v) = z(F^i)$ holds that $v + k \cdot \DSW$ is an empty node for all $k \geq 0$.
    Otherwise there would exist a tiled node whose $xy$-coordinate is not covered by some node of $F^i$, which would contradict that $F^i$ satisfies \propertyTwo{}.
    This observation directly implies that for every node $u \in U^i$, there exists a corresponding node $w \in F^i$ such that $xy(u) = xy(w)$.
    Combining this with the previously established relation $z(u^i) \leq z(F^i)$, we can conclude that $z(u) < z(F^i)$ holds for all nodes $u \in U^i$.

    Consider $w = u^i + k \cdot \UNE$, a node within $F^i$ with $xy(w) = xy(u^i)$.
    Given that $u^i$ is a node with the maximum $z$-coordinate in $U^i$, we can deduce that $u^i + j \cdot \UNE$ is a tiled node for all $k \geq j > 0$.
    The proof of \cref{lem:propertiesMaintained} establishes that if $F^i$ satisfies \allProperties{}, then it undergoes projection between steps $i$ and $i+1$.
    As a result of that projection, a specific tile shift is performed, precisely from node $w$ to node $u^i$.
    Consequently, this tile shift implies that $u^i$ is no longer part of $U^{i+1}$, which concludes the first part of the proof.

    To prove that a projection cannot introduce any nodes to $U^{i+1}$, we begin by considering an arbitrary node $u^{i+1} \in U^{i+1}$ and assume, for contradiction, that $u^{i+1} \notin U^i$.
    In step $i+1$, as per definition of $U^{i+1}$, there exists a tiled node $v = u^{i+1} + k \cdot \DSW$ for some $k > 0$.
    The absence of $u^{i+1}$ in $U^i$ leads to two possible cases: either (1) $u^{i+1}$ was not empty in step $i$, or (2) $v$ was not tiled in step $i$.

    In the first case, since the only tiles picked up between steps $i$ and $i+1$ are at nodes of $F^i$, we conclude that $u^{i+1} \in F^i$.
    By definition, $u^{i+1} \in \bb(\tiled^{i+1})$, implying the existence of a tiled node $w$ in step $i+1$ with $z(w) \geq z(F^i)$.
    However, in step $i$, there cannot be a tiled node with $z$-coordinate greater than $z(F^i)$.
    The fact that each node in $F^i$ is empty in step $i+1$ leads to the contradiction that, in step $i$, there are two distinct fragments with $z$-coordinate $z(F^i)$, violating \propertyTwo{}.

    In the second case, between steps $i$ and $i+1$, a tile must have been shifted from some node $w \in F^i$ to $v = u^{i+1} + k \cdot \DSW$ where $v$ is the first empty node $\DSW$ of node $w$.
    This implies $xy(w) = xy(u^{i+1})$.
    Similar to the first case, in step $i$, no node in $\bb(\tiled^{i})$ has a $z$-coordinate greater than $z(F^i)$.
    In step $i+1$, each node in $F^i$ is empty, leading to the conclusion that every node in $\bb(\tiled^{i+1})$ must have a $z$-coordinate smaller than $z(F^i)$.
    Combining this with $xy(w) = xy(v) = xy(u^{i+1})$, we derive $z(v) < z(u^{i+1}) < z(w)$.
    This contradicts that $v$ must be the first empty node $\DSW$ of $w$.

    Both cases lead to a contradiction, implying that a node $u^{i+1} \in U^{i+1}$ with $u^{i+1} \notin U^i$ does not exist.
    This completes the second part of the proof and, consequently, the lemma.
\end{proof}

Once the agent enters a configuration where it is positioned within a fragment satisfying \allProperties{}, it consistently remains within such a fragment in subsequent configurations according to \cref{lem:propertiesMaintained}.
By \cref{lem:convergenceIntermediate}, such a configuration must eventually be reached, and by the previous lemma, our second potential $\Psi^i$ is strictly monotonically decreasing afterwards.
It follows that the set of empty nodes within $\bb(\tiled{}^i)$ that have a tile somewhere in the $\DSW$ direction must eventually be empty.
Consequently, any tiled node within the bounding box that is not contained in the singular uppermost fragment satisfying \allProperties{} must possess a neighboring tile at $\UNE$.
Hence, the entire configuration satisfies \cref{def:icicle}, which is captured by the following theorem, serving as the conclusion of our analysis:

\begin{theorem}\label{thm:convergence}
    The sequence of configurations resulting from the execution of \icicleAlgo{} on any initially connected configuration $C^0 = (\tiled^0,p^0)$ with $p^0 \in \tiled^0$ converges to an icicle.
\end{theorem}

\section{Termination Criteria and Runtime}\label{sec:runtime}

Once the agent is positioned within a fragment satisfying \allProperties{}, it remains within such a fragment and subsequently exclusively performs projections.
In the case where the configuration is already an icicle (see \cref{def:icicle}), every tiled node in the configuration must be traversed during these projections.
This condition is essential for our termination check.
The agent maintains a flag \emph{term}, which it flags as true upon initiating a projection.
This flag only reverts to false if the agent detects any violation of \cref{def:icicle} during the ongoing projection, i.e., whenever a tiled node $v$ is observed for which $v + \UNE \notin F$ and $v + \UNE \notin \tiled{}$, where $F$ is the fragment in which the projection was initiated.
Once \emph{term} still holds after a projection, the agent terminates.
Note that the flag is reverted, even if $v + \UNE$ is tiled immediately afterwards.
As an example, if a tile shift from some node $w \in F$ to $v + \UNE$ is performed as part of a projection, then node $v$ is observed before the tile is placed at $v + \UNE$.
Although it is possible that the configuration is an icicle after tiling node $v + \UNE$, the agent cannot verify it during that projection, as it does not traverse node $v$ or any node $\DSW$ of $v$.

In general, by adhering to this termination procedure, the agent consistently performs one additional projection once the configuration converges to an icicle.
Since the algorithm only terminates following a projection in which it could observe all tiled nodes and only if, in this case, \cref{def:icicle} is satisfied, the correctness of our algorithm is established.

The algorithm's runtime can be expressed as the sum $t_{total} = t_{proj} + t_{shift} + t_{move}$, where $t_{proj}$ accounts for all steps performed during the projection subroutine, $t_{shift}$ for steps that are performed as part of some tile shift (outside of a projection), and $t_{move}$ for any remaining step.
We bound each term individually by $\O(n^3)$, which gives a runtime of $\O(n^3)$ in total.

\begin{lemma}
    \label{lem:runtimeProjection}
    The total number of steps performed during projections is $t_{proj} = \O(n^3)$.
\end{lemma}

\begin{proof}
    Consider the projection of some fragment $F$.
    During that projection, the agent visits each node $v \in F$ at most three times: at most once moving $\N$ to find the $\N$-most node at which the projection of a column is initiated, once moving $\S$ during the projection of that column, and once to pick up the tile at $v$ after $v$ is projected.
    Any tiled node $v \in \tiled \setminus F$ is visited at most twice: at most once moving $\DSW$ and at most once moving back $\UNE$.
    In total, all tiled nodes are visited at most three times, which implies that the projection of $F$ takes $\O(n)$ steps.
    In the following, we show that no more than $\O(n^2)$ projections are performed throughout the execution.

    Consider the projection of any fragment $F$ prior to the existence of a fragment that satisfies \allProperties{}.
    Note that $F$ must be a platform in order to be projected.
    We distinguish whether $F$ is a platform that results from the execution of \cref{case:d}, in which case we call $F$ a \emph{temporary} platform, or whether $F$ was already present initially, in which case we call $F$ an \emph{initial} platform.
    There are at most $\O(n)$ initial platforms, each of which requires at most $\O(n)$ projections until the number $|\platforms{}|$ of platforms reduces by one, e.g., if all nodes $v + k \cdot \DSW$ are occupied for all nodes $v \in F$ of some initial platform $F$ and all $1 \leq k \leq \O(n)$.
    It follows that at most $\O(n^2)$ projections of initial platforms can occur.
    Now we consider temporary platforms.
    Denote $v_0, ..., v_k$ the nodes whose tiles are shifted to nodes $w_0,..., w_k$ during the execution of \cref{case:d} (see \cref{subfig:criticalCases-S}).
    Let $F^*$ be the fragment that contains $v_0,...,v_k$ prior to the tile shifts, and $F_1, F_2$ the resulting fragments after the tile shift, i.e., $F^* \setminus \{v_0,...,v_k\} = F_1 \cup F_2$.
    The number of platforms increases, if either $F^*$ is already a platform, or the nodes in $F^*$ with an upwards neighboring tile are all contained in $F_1$ or all contained in $F_2$.
    As a result, the execution of \cref{case:d} can increase the number of platforms at most by one.
    W.l.o.g. let $F_1$ be projected before $F_2$.
    After the projection of $F_1$, all nodes in $\{v + \DSW \mid v \in F_1\}$ are contained in the same fragment as the nodes $w_0,...,w_k$, and that fragment is adjacent to $F_2$.
    This implies, that a single projection suffices to reduce the number of platforms by one whenever $F_1$ or $F_2$ is temporary.
    Whenever a temporary platform is created in \cref{case:d}, the $x$-coordinate of at least two tiles decreases by one.
    By \cref{lem:maintainBoundingCylinder}, the $x$-coordinate cannot be negative, which implies that it can decrease by at most $n$ for any tile.
    Thereby, \cref{case:d} can be executed at most $\O(n^2)$ times, which adds $\O(n^2)$ temporary platforms and thereby $\O(n^2)$ additional projections.
    This concludes that the total number of projections prior to the existence of a fragment that satisfies \allProperties{} is $\O(n^2)$.

    Once a fragment that satisfies \allProperties{} exists, at most $n$ additional projections are required before the configuration is an icicle, i.e., after at most $n+1 = \O(n)$ additional projections the agent terminates.
    This follows directly from the observation that the bounding boxes of tiled nodes before and after the $n+1$ projections are completely disjoint, and since we cannot create any node that violates the shape of the icicle during those projections as shown in the proof of \cref{lem:potential2Monotonicity}.

    We have shown that $\O(n^2)$ projections are performed in total, each of which requires $\O(n)$ steps.
    From this, the lemma follows directly.
\end{proof}

\begin{lemma}
    \label{lem:runtimeOther}
    Let $i$ be the first step following an arbitrary projection, and $j > i$ the next step in which a projection is initiated.
    Between step $i$ and $j$ at most $\O(n)$ steps are performed outside of tile shifts, and any tile shift from some node $v$ to $w$ takes $\O(xy(v)-xy(w))$ steps.
\end{lemma}

\begin{proof}
    We begin by addressing the latter statement of the lemma.
    Consider the number of steps performed during a tile shift from some node $v$ to $w$.
    The \buildPar{} procedure lets the agent visit nodes in decreasing order of their $xy$-coordinates.
    Nodes $u$ with $xy(v) > xy(u) > xy(w)$ are visited at most twice before placing a tile at node $w$, and at most twice again before picking up the tile at $v$.
    Additionally, two steps are required for both tile placement and retrieval.
    Consequently, the claim that the tile shift takes $\O(xy(v)-xy(w))$ steps directly follows for tile shifts during procedure \buildPar{}.

    Moving on to \cref{case:a,case:b,case:c}, a single tile is shifted in the directions $\SE, \DE$, or $\SE + \DSW$, requiring $\O(1)$ steps.
    However, as per \cref{eq:transition}, the $x$-coordinate of the tile decreases by one, implying a $\Omega(h)$ reduction in the tile's $xy$-coordinate, where $h$ is the height of the bounding cylinder of $\tiled{}$.
    This observation becomes useful below, where we use the difference between the performed $\O(1)$ steps and the reduction of the $xy$-coordinate by $\Omega(h)$ to account for the additional movement that follows \cref{case:b}.
    In the case of tile shifts in \cref{case:d,case:e}, the agent follows a two-step process: initially traversing $k$ tiles of the column in the $\S$ direction and subsequently shifting all $k$ tiles in the $\DE$ direction.
    The proportional number of steps per shifted tile is constant which concludes the latter statement of the lemma.

    Next, consider any move between step $i$ and $j$ that is not part of some tile shift.
    In step $i$, the agent repeatedly moves $\UW,\USE,\UNE,\NW,\SW$ or $\N$ in the search phase.
    Notably, none of these moves can decrease its $x$-, $y$-, or $z$-coordinate.
    Consequently, after $\O(n)$ steps, the agent must leave the search phase at some northernmost node $v$ of some column $c$.

    Subsequently, whenever the agent finds a column with the same $x$-coordinate as $c$ to the north while moving $\N$ in the column $c'$ east of $c$, it enters the northernmost adjacent column $c^*$ west of $c'$ and transitions to the search phase.
    In such instances, no tile in column $c^*$ was visited between step $i$ and $j$.
    Consequently, the tiles in column $c'$ that were visited in order to find $c^*$ did not contribute to the $\O(n)$ steps of the initial search phase, allowing us to include the additional steps taken to visit $c^*$ in the overall count.

    Apart from the initial search phase and the first visits to more northern columns of equal $x$-coordinate, the agent moves without shifting a tile in the following scenarios:
    \begin{enumerate}
        \item Whenever it encounters a layer above in the \buildPar{} procedure or \cref{case:d,case:e}.
        \item Whenever it encounters a more western column (in the same procedure and cases).
        \item After picking up a tile at $v$, where $v + \S \notin \tiled$, i.e., after removing the last tile of column $c$. In this case, the agent moves an additional $\O(|c'|)$ steps within the eastern neighboring column $c'$ of $c$.
        \item After picking up a tile at $v$ where $v + \S \in \tiled$, i.e., $v$ is not the last tile of column $c$. In this case, the agent moves a single step $\S$.
    \end{enumerate}
    First, we consider case 1, where the agent discovers an upper layer.
    Notably, outside of the projection subroutine, the agent never moves to a layer below.
    Given our assumption that no projection occurs between step $i$ and $j$, it follows that each time the agent moves to an upper layer in the \buildPar{} procedure or in \cref{case:d,case:e}, the tiles visited in that particular execution of \buildPar{} or \cref{case:d,case:e} are not revisited until after step $j$.
    This implies that all tiles can be visited only a constant number of times over all these executions, thereby giving an overhead of $\O(n)$ steps for case 1.

    Similarly, excluding tile shifts, the agent exclusively moves to the east after completely removing all tiles from a column $c$.
    We argue that in such instances, the column $c'$ east of $c$ cannot have any more neighboring tiles to the west.
    Moreover, except for the initial search phase in step $i$, more western columns are only found following the additional column traversal in \cref{case:b}.
    Afterwards, we account for the overhead steps resulting from that additional column traversal using the constant runtime for the tile shift in \cref{case:b}, as explained in the first paragraph of this proof.
    Consequently, case 2 does not introduce any additional steps that are not already accounted for by the constant runtime tile shifts.

    To elaborate further, let's focus on case 3.
    There are two scenarios in which the agent picks up the tile at the last node $v$ of some column $c$.
    Firstly, during the execution of \cref{case:b}, after which the eastern neighboring column $c'$ of $c$ is traversed in the $\S$ direction, either leading the agent to a western neighboring column of $c'$ after $\O(|c'|)$ steps, or positioning it at the northernmost tile of column $c'$ if no such column is found.
    We can count the $\O(|c'|) = \O(h)$ steps that follow the execution of \cref{case:b} to the constantly many steps required for \cref{case:b} (which we have shown in the first paragraph of this proof).
    The second scenario is after shifting a tile in the \buildPar{} procedure.
    In this case, node $v + \SE$ must be the southernmost node of column $c'$.
    We can further distinguish two cases: (1) the tile at $v$ is shifted to the empty node $w$ that is north of column $c'$, or (2) the tile is shifted to some node $w$ further east than column $c'$.
    Note that in both cases, $c'$ cannot have any further neighboring column to the west, as otherwise, the agent would not have fully traversed column $c'$ in the $\N$ direction.
    In the former case, the tile at node $w$ is shifted to some node $u$ afterward.
    Node $u$ is either located east of column $c'$ or directly south of $c'$.
    In either case, it holds that $xy(w) - xy(u) > |c'|$, and we can count the $\O(|c'|)$ steps to the subsequent tile shift from $w$ to $u$.
    In the latter case, we can count the $\O(|c'|)$ steps directly to the tile shift from $v$ to $w$, as in this case $xy(v)-xy(w) > |c'|$.

    To summarize, each overhead step resulting from case 3 is effectively covered by the constant runtime of the tile shift in \cref{case:b}.
    Whenever a more western column is found outside of the additional column traversal that follow \cref{case:b}, that column is visited for the first time between step $i$ and $j$, in which case we can count the steps necessary to reach that column to the initial $\O(n)$ steps in phase search.
    Hence case 2 does not produce additional overhead.

    Finally, we simply account for the single step from node $v$ to node $v + \S$ following case 4 as part of the preceding tile shift.
\end{proof}

As established in the proof of \cref{lem:runtimeProjection}, a total of $\O(n^2)$ projections are performed., which together with \cref{lem:runtimeOther} implies that $t_{other} = \O(n^3)$.
The $xy$-coordinate of any tile is at most $n^2$, non-increasing, and cannot be negative (see \cref{lem:potential1Monotonicity,lem:maintainBoundingCylinder}).
Together with \cref{lem:runtimeOther}, each tile contributes $\O(n^2)$ steps to $t_{shift}$, which implies that $t_{shift} = \O(n^3)$.
This concludes our final theorem:

\begin{theorem}
    \label{thm:icicleRuntime}
    \icicleAlgo{} has a runtime of $\O(n^3)$ steps.
\end{theorem}

\section{Experimental Analysis}\label{sec:experimental}
\begin{figure}[!t]
    \centering%
    \includegraphics[width=\linewidth]{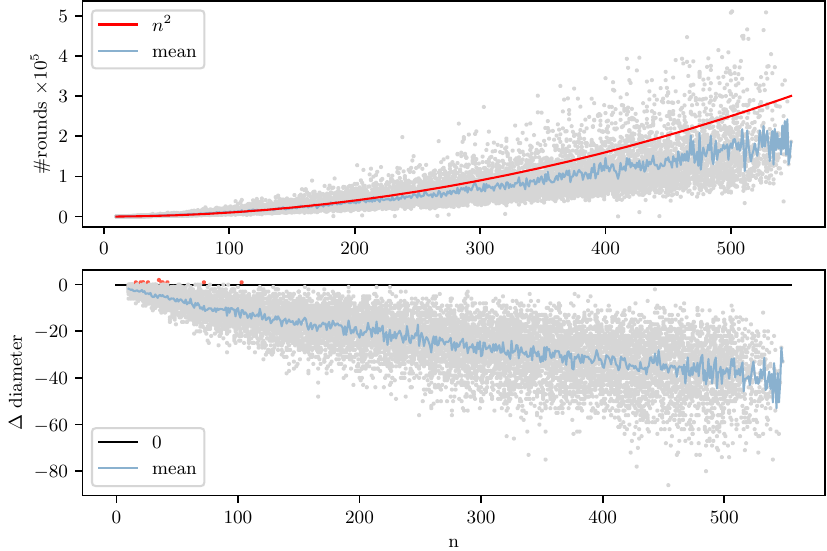}%
    \caption{
        The results stem from 12,250 simulations involving random configurations ranging in sizes from 10 to 550. The upper plot shows the number of steps until termination. The lower plot shows the difference in diameter between the input and output configurations.}
    \label{fig:plots}
\end{figure}
While our algorithm matches the runtime bound of the 3D line formation algorithm \cite{Hinnenthal2020Hybrid3D}, the icicle offers distinct advantages over the line.
The diameter of an icicle can be as low as $\O(n^{\frac{1}{3}})$, whereas a line consistently maintains a diameter of $n$.
Unfortunately, our algorithm does not improve the diameter if the initial configuration already closely resembles a line.
On the other hand, we conjecture that if the initial diameter is as low as $\O(n^{\frac{1}{3}})$ (which is the best case in 3D), then our algorithm can only increase the diameter by a factor of $\O(n^{\frac{1}{3}})$.
We support our conjecture with the following simulation results on configurations where initially all tiles are contained in a sphere of radius $\O(n^{\frac{1}{3}})$.
We conducted a total of 12,250 simulations using the icicle formation algorithm on random configurations. 
For each value of $n$ within the range $10 \leq n < 500$, we sampled 25 random configurations as follows: empty nodes were repeatedly chosen uniformly at random within a sphere of radius $4n^{\frac{1}{3}}$, and a tile was placed on each selected node until a connected component of tiled nodes with a size of at least $n$ was formed. 
Subsequently, any tile outside of that component was removed, the agent was placed at a randomly chosen tile within the component, and the algorithm was simulated until termination.
We measured the runtime as well as the difference in diameter, which are plotted in Figure \ref{fig:plots}.
Due to the nature of the described random generation process, configurations of size larger than 500 were also sampled, although less frequently.
Specifically, we observed an average sampling rate of approx. 24.4 configurations for sizes at most 450, contrasting with approx. 15.4 configurations for sizes exceeding 450. 
This discrepancy contributes to the noticeable increase in variance as the configuration size approaches the 500 threshold.
The runtime stays below the proven upper bound of $\O(n^3)$, and notably, it remains well in the vicinity of $n^2$.
This can be attributed to the initial close packing of all tiles in our random configurations. 
Additionally, instances where the diameter increases (indicated by red dots above the $x$-axis) are infrequent, and their occurrence diminishes as the configuration size increases.
This trend implies a general decrease in diameter in the average case.
In the following, we present what we believe to be the worst-case configuration.
Consider the configuration $C$ depicted in \cref{subfig:worstCase1} that consists of three layers.
The middle layer contains $k = \Theta(n^{\frac{1}{3}})$ fragments $F_1,...,F_k$ ordered from east to west, where each $F_i$ has size $\O(i)$ and the agent's initial position is $p^0 \in F_1$.
It further contains a fragment $F_0$ of size one east of the agent's initial position.
Observe that the bounding box of $F_i$ contains no node from $F_{i+1}$ for any $i$ with $0 < i < k$.
It follows that the agent builds and projects parallelograms in the order $F_1,...,F_k$.
Since the bounding box of $F_i$ contains $p^0$ for all $i > 0$, it further follows that $k$ tiles are projected from $p^0$ in direction $\DSW$.
Only then, the agent traverses the lower layer and eventually finds fragment $F_0$ where it moves further upwards.
Now consider the configuration that consists of $\Theta(n^{\frac{1}{3}})$ copies of $C$ in direction $\UNE$ (see \cref{subfig:worstCase2}).
That configuration has diameter $\O(n^{\frac{1}{3}})$ initially.
Throughout the icicle formation algorithm, some tile at node $p^0$ is projected $\Theta(n^{\frac{2}{3}})$ times, which implies that the resulting icicle has depth and thereby also diameter $\Theta(n^{\frac{2}{3}})$.
\begin{figure}[t]
    \centering
    \begin{subfigure}[b]{0.5315\linewidth}
        \centering%
        \includegraphics[width=0.9\linewidth,page=1]{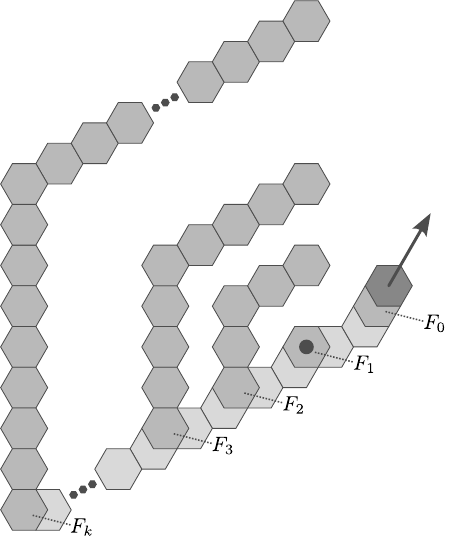}%
        \subcaption{}
        \label{subfig:worstCase1}
    \end{subfigure}%
    \begin{subfigure}[b]{0.4685\linewidth}
        \centering%
        \includegraphics[width=0.9\linewidth,page=2]{worstCase}%
        \subcaption{}
        \label{subfig:worstCase2}
    \end{subfigure}%
    \caption{Illustrating what we believe to be the worst case configuration in terms of increase in diameter. In (a), the three lowest layers of the configuration are depicted in detail. The second-lowest layer contains $k = \Theta(n^{\frac{1}{3}})$ fragments $F_i$, each of size $\O(i)$, and the agent's initial position $p^0 \in F_1$. Observe that the bounding box of each $F_i$ contains $p^0$. The whole configuration is depicted in (b) and consists of $\Theta(n^{\frac{1}{3}})$ copies of the layers depited in (a) in direction $\UNE$ (indicated by the arrows).}
    \label{fig:worstCase}
\end{figure}
\section{Future Work}\label{sec:future}

In this work, we introduced an algorithm capable of transforming any initially connected configuration into an icicle within $\O(n^3)$ steps, complemented by proofs of correctness and runtime analysis. 
While our algorithm's experimental results are promising, future work should include a formal proof to substantiate the claimed upper bound of $\O(n^{\frac{1}{3}})$ on the increase in diameter.
Additionally, the adaptability of our algorithm to the multi-agent case poses an intriguing challenge for future investigation. 
Given that the algorithm comprises distinct phases potentially executed in an interleaved manner, addressing its integration into a multi-agent framework presents a non-trivial research direction. 

\bibliography{bibliography}
  \newpage
\appendix
\let\oldnl\nl
\newcommand{\nonl}{\renewcommand{\nl}{\let\nl\oldnl}}
\SetKwRepeat{Do}{do}{while}
\newcommand{\lIfElse}[3]{\lIf{#1}{#2 \textbf{else}~#3}}

\section{Deferred Pseudocode}\label{sec:appendix-pseudocode}

\begin{algorithm}[!b]
    \caption{\textsc{2DParallelogramFormation}}
    \label{alg:parallelogram}
     \While{true}{
         \lWhile{$\{p + \NW,p + \SW,p + \N\} \cap \tiled{} \neq \emptyset$}{
             move to tile at $\NW, \SW$ or $\N$ 
        }
         $firstColumn \leftarrow$ true; run \buildPar{}\;
         \lIf(\Comment*[f]{terminate $\S$ of easternmost column}){$p \notin \tiled{}$}{
            \Return
        }
         \ElseIf{$r$ carries no tile}{
             \uIf{$firstColumn$}{
                 \lWhile{$p + \N \in \tiled{}$}{move $\N$}
            }
             \Else{
                 \lWhile{$\{p + \SW,p + \S\} \cap \tiled{} \neq \emptyset$}{
                     move to tile at $\SW$ or $\S$
                }
                 \lWhile{$\{p + \NW,p + \SW,p + \N\} \cap \tiled{} \neq \emptyset$}{
                     move to tile at $\NW, \SW$ or $\N$
                }
            }
             pickup tile; move to tile at $\S, \SE$ or $\NE$ \;
        }
    }
    \BlankLine
    \nonl\textbf{procedure} \buildPar\;
     \While{$p\in T$}{
         \uIf(\Comment*[f]{irrelevant in 2D}){$p + \UW \in \tiled{}$ or $p + \USE \in \tiled{}$ or $p + \UNE \in \tiled{}$ }{
             move to tile at $\UW, \USE$ or $\UNE$ ; \Return\;
        }
         \uElseIf{$firstColumn$ \textbf{and} $p + \SW \in \tiled{}$}{
             move $\SW$; \Return
            \Comment*[r]{found more western column}
        }
         \uElseIf{$p + \NE \in \tiled{}$ \textbf{and} $p + \SE \notin \tiled{}$}{
             move $\SE$; place tile; move $\NW$; \Return
            \Comment*[r]{place tile below eastern column}
        }
         \ElseIf{$p + \N, p + \SE \in \tiled{}$ \textbf{and} $p + \NE \notin \tiled{}$}{
             move $\NE$; place tile; move $\SW$; \Return
            \Comment*[r]{place tile above eastern column}
        }
         move $\S$\;
    }
     \uIf{$p + \N, p + \NE, p + \SE \in \tiled{}$}{
         place tile; move $\N$
        \Comment*[r]{place tile below current column}
    }
     \ElseIf{$p + \NE \in \tiled{}$}{
         move $\NE$; move $\N$; $firstColumn \leftarrow false$
        \Comment*[r]{move to top of next column}
         \While{$p \in \tiled{}$}{
             \If{$p + \SW \notin \tiled{}$ and $p + NW \in \tiled{}$}{
                 \lWhile{$p + \N \in \tiled$}{move $\N$}
                 \lWhile{$p + \NW \notin \tiled$}{move $\S$}
                 \Return  \Comment*[r]{found more northern column}
            }
             move $\N$
        }
         \lIf(\Comment*[f]{place tile above current column}){$p + \S, p + \SE \in \tiled{}$}{
            place tile
        }
         \lElse{
            move $\S$; run \buildPar{}
        }
    }
     \Return

\end{algorithm}

The pseudocode for the 2D parallelogram formation algorithm, as detailed in \cref{subsec:algorithm-parallelogram}, is given by \cref{alg:parallelogram}.
Specifically, lines 12--34 within \cref{alg:parallelogram} describe the \buildPar{} procedure, utilized by both the parallelogram and icicle formation algorithms.
Note that the checks for tiles above (lines 13--14) can be disregarded in the 2D setting, as they only become relevant in the icicle formation algorithm.
The \proj{} procedure is given in \cref{alg:project}, and the full icicle formation algorithm in \cref{alg:icicle}.
Whenever multiple directions of movement are specified, their precedence is implicit in the provided order.

\begin{algorithm}[!t]
    \caption{\icicleAlgo}
    \label{alg:icicle}
     \While{true}{
         \While{$\{p+ \X \mid \X \in \{\UW,\USE,\UNE,\NW,\SW,\N\}\} \cap T \neq \emptyset$}{
             move to tile at $\UW, \USE, \UNE, \NW, \SW$ or $\N$ \;
        }
         \uIf{$G|_{\tiledNeighb{}(p)}$ or $G|_{\tiledNeighb{}(p)\cup\{p + \SE\}}$ is connected \textbf{ or } $G|_{\tiledNeighb{}(p)\cup\{p + \SE + \DSW\}}$ is connected with $p + \SE + \DSW \in \tiled{}$}{
             $firstColumn \leftarrow$ true; run \buildPar{}\;
             \lIf{$p \notin \tiled{}$}{
                move $\N$; run \proj{}
            }
             \lElseIf(\Comment*[f]{same as lines 8--15 from \cref{alg:parallelogram}}){$r$ carries no tile}{...}
            \setcounter{AlgoLine}{15}
        }
         \uElseIf{$G|_{\tiledNeighb{}(p)\cup\{p + \DE\}}$ is connected}{
             \uIf{$\tiledNeighb{}(p) = \{p+\DSW, p + \SE\}$}{
                 move $\SE + \DSW$; place tile; move $\UNE + \NW$; pickup tile; move $\SE$ \;
            }
             \lElse{move $\DE$; place tile; move $\UW$; pickup tile}

             \uIf{$p + \SE, p + \NE \in \tiled{}$ and $p + \S \notin \tiled{}$}{
                 move $\SE$; \lWhile{$\{p + \UW, p + \USE, p + \SW, p + \S\} \cap \tiled{} = \{p + \S\}$}{move $\S$}
            }
             \lElse{move to tile at $\S, \SE$ or $\NE$ }

        }
         \uElseIf{$\tiledNeighb{}(p) = \{p + \DNW, p + \S, p + \NE\}$}{
             \lWhile{$\{p + \X \mid \X \in \{\UW,\USE, \SW, \DSW, \SE, \DE, \S\}\} \cap \tiled{} = \{p + \S\}$}{move $\S$}
             \If{$\{p+ \X \mid \X \in \{\UW, \USE, \SW\}\} \cap \tiled{} = \emptyset$}{
                 \lIf{$p + \DE \in \tiled{}$}{move $\N$}
                 move $\DE$; place tile; move $\UW$; pickup tile\;
                 \lWhile{$p + \N \in \tiled$}{move $\SE + \DNW$; place tile; move $\UW$; pickup tile}
                 move $\NE$\;
            }
        }
         \ElseIf{$\tiledNeighb{}(p) = \{p + \DNW, p + \S\}$}{
             \lWhile{$\{p + \X \mid \X \in \{\UW,\USE, \SW, \SE, \DE, \S\}\} \cap \tiled{} = \{p + \S\}$}{move $\S$}
             \If{$\{p+ \X \mid \X \in \{\UW, \USE, \SW\}\} \cap \tiled{} = \emptyset$}{
                 \lIf{$\{p+ \X \mid \X \in \{\SE, \DE\}\} \cap \tiled{} = \emptyset$}{move $\N$; run \proj{}}
                 \Else{
                    ... \Comment*[r]{same as lines 26--28}
                    \setcounter{AlgoLine}{37}
                     \While{$p \notin \tiled$}{
                         move $\S$; \lIf{$p + \SE \in \tiled{}$}{move $\SE$}
                         
                    }
                }
            }
        }
    }
\end{algorithm}

In \cref{alg:parallelogram}, the agent traverses a column in the $\S$ direction in lines 12--21 and the next column in the $\N$ direction in lines 26--31.
Following the check for whether the empty node above the next column should be tiled (line 32), the agent recursively executes \buildPar{} starting at the $\N$-most node of the next column (line 33).
The procedure may return with the agent being in various states, such as positioned on a tiled or empty node, with or without a tile.
In the main loop of the algorithm, \buildPar{} is executed repeatedly, and distinctions between these states are made to either terminate (line 4), retrieve the tile at which \buildPar{} was previously entered (lines 5--11), or enter the search phase (line 2).

In \cref{alg:icicle}, lines 16--22 are dedicated to handling \cref{case:b}, where a tile is shifted in the $\DE$ direction to maintain connectivity.
Lines 23--29 cover \cref{case:d}, and lines 30--40 cover \cref{case:e}.
These cases involve shifting multiple tiles of a column in the $\DE$ direction.
For concise pseudocode, the handling of \cref{case:a} is delegated to the \buildPar{} procedure.
This is because, in that procedure, $v + \SE$ is tiled in that particular case.
Additionally, within this context, we have integrated the check for a tile at $v + \SE + \DSW$ from \cref{case:c} (refer to lines 4–-5).

\begin{algorithm}[!ht]
    \caption{}
    \label{alg:project}
    \nonl\textbf{procedure} \proj\;
     \uIf(\Comment*[f]{parallelogram of height one}){$p + \N, p + \S \notin \tiled$}{
         \Do{$p + \UNE \in \tiled{}$}{
             \lWhile{$p \in \tiled{}$}{move $\DSW$}
             place tile; \lWhile{$p + \UNE \in \tiled{}$}{move $\UNE$}
             pickup tile\;
             \lIfElse{$p + \NW \in \tiled{}$}{move $\SW$; move $\DNW$}{move $\DSW$; \Return}
        }
    }
     \Else{
         \Do{$p \in \tiled{}$}{
             \lWhile{$p + \N \in \tiled{}$}{move $\N$}
             \lWhile{$p \in \tiled{}$}{move $\DSW$}
             place tile; \lWhile{$p + \UNE \in \tiled{}$}{move $\UNE$}
             pickup tile\;
             \lIf{$p + \S \in \tiled{}$}{move $\S$}
             \lElseIf{$p + \NW \in \tiled{}$}{move $\NW$}
             \lElse{move $\DSW$; \Return}
        }
    }
\end{algorithm}

\end{document}